\documentclass[10pt,reqno]{amsart}
\usepackage{amsthm,amsfonts,amsmath,latexsym,amssymb,bbold,dsfont}
\usepackage{amssymb,latexsym,amsfonts,amsmath,amsthm,mathabx}
\usepackage{hyperref,verbatim, eucal}

\usepackage[dvips]{graphics}
\usepackage{graphicx,color}
\usepackage{epsfig}
\usepackage{cite}
\def\a{\alpha}
\def\b{\beta}
\def\g{\gamma}
\def\d{\delta}
\def\e{{\rm e}}

\def\l{\lambda}
\def\L{\Lambda}
\def\p{\partial}

\def\R{{\mathbb R}}

\def\N{\mathbb N}
\def\tr{{\rm tr}\,}

\theoremstyle{plain}
\newtheorem{thm}{THEOREM}
\newtheorem{lemma}[thm]{LEMMA}

\newtheorem{prop}[thm]{PROPOSITION}

\theoremstyle{definition}

\theoremstyle{remark}

\newtheorem{remarks}[thm]{REMARKS}

\newcommand{\be}{\begin{equation}}
\newcommand{\ee}{\end{equation}}
\newcommand{\bea}{\begin{eqnarray}}
\newcommand{\eea}{\end{eqnarray}}
\newcommand{\beax}{\begin{eqnarray*}}
\newcommand{\eeax}{\end{eqnarray*}}

\newcommand{\mfr}[2]{{\textstyle\frac{#1}{#2}}}

\numberwithin{equation}{section}

\begin{document}

\title[Local ground-state entropy of free fermions in a constant magnetic field]{Asymptotic growth of the local ground-state entropy of the ideal Fermi gas in a constant magnetic field}
\date{November 13, 2020}

\author[H.~Leschke, A.V.~Sobolev, W.~Spitzer]
{Hajo Leschke, Alexander V.~Sobolev, Wolfgang Spitzer}
\address{Institut f\"ur Theoretische Physik, 
Universit\"at Erlangen-N\"urnberg, 
Staudtstra\ss e 7, 91058 Erlangen, Germany}
\email{hajo.leschke@physik.uni-erlangen.de}
\address{Department of Mathematics\\ University College London\\
Gower Street\\ London\\ WC1E 6BT UK}
\email{a.sobolev@ucl.ac.uk}
\address{Fakult\"at f\"ur Mathematik und Informatik, 
FernUniversit\"at in Hagen, Universit\"atsstra\ss e 1, 
58097 Hagen, Germany}
\email{wolfgang.spitzer@fernuni-hagen.de}

\begin{abstract} We consider the ideal Fermi gas of indistinguishable particles without spin but with electric charge, confined to a Euclidean plane $\R^2$ perpendicular to an external constant magnetic field of strength $B>0$. We assume this (infinite) quantum gas to be in thermal equilibrium at zero temperature, that is, in its ground state with chemical potential $\mu\ge B$ (in suitable physical units). For this (pure) state we define its local entropy $S(\Lambda)$ associated with a bounded (sub)region $\Lambda\subset \R^2$ as the von Neumann entropy of the (mixed) local substate obtained by reducing the infinite-area ground state to this region $\Lambda$ of finite area $|\Lambda|$. In this setting we prove that the leading asymptotic growth of $S(L\Lambda)$, as the dimensionless scaling parameter $L>0$ tends to infinity, has the form  $L\sqrt{B}|\p\L|$ up to a precisely given (positive multiplicative) coefficient which is independent of $\L$ and dependent on $B$ and $\mu$ only through the integer part of $(\mu/B-1)/2$. Here we have assumed the boundary curve $\p\L$ of $\Lambda$ to be sufficiently smooth which, in particular, ensures that its arc length $|\p\L|$ is well-defined. This result is in agreement with a so-called area-law scaling (for two spatial dimensions). It contrasts the zero-field case $B=0$, where an additional logarithmic factor $\ln(L)$ is known to be present. We also have a similar result, with a slightly more explicit coefficient, for the simpler situation where the underlying single-particle Hamiltonian, known as the Landau Hamiltonian, is restricted from its natural Hilbert space $\text L^2(\R^2)$ to the eigenspace of a single but arbitrary Landau level. Both results extend to the whole one-parameter family of quantum R\'enyi entropies. As opposed to the case $B=0$, the corresponding asymptotic coefficients depend on the R\'enyi index in a non-trivial way.
\end{abstract}

\maketitle
\tableofcontents
\section{Introduction}

Quantum correlations in many-particle ground states occur in a genuine and simple form for fermions without interactions between them. In this case all correlations are exclusively due to the Pauli--Fermi--Dirac statistics and are not affected by classical correlations. This certainly explains why the authors of many recent publications, devoted  to the ``trendy topic'' of entanglement entropy, have considered ground states of free fermions in discrete or continuous position space. For these pure states its (bipartite spatial) entanglement entropy boils down to its local entropy associated with a bounded region $\Lambda$ in the position space. An informal definition of the local entropy is given in the above abstract. For a formal definition, in the present context, see \eqref{entropy} (with $\a=1$) below. This local ground-state entropy may serve as a useful, but rough, single-number quantification of the correlations of all the particles in the region $\Lambda$ with all those outside.

The local ground-state entropy is a complicated function(al) of $\Lambda$ and difficult to study by analytic methods. Even without interactions, one can in general only hope for estimates and/or asymptotic results when the volume of the bounded region becomes large. As discovered by Gioev and Klich\cite{GK}, one fascinating aspect of these type of asymptotic results is the connection to the quasi-classical evaluation of traces of (truncated) Wiener--Hopf operators (or Toeplitz matrices in the discrete, one-dimensional, case), that is, to a conjecture of Harold Widom (respectively of Fisher and Hartwig). The ``Widom conjecture" was finally proved by one of us in \cite{Sob:AMS} and opened the gate to prove a conjecture by Gioev and Klich \cite{GK} about the precise asymptotic growth of the local ground-state entropy of free fermions in multi-dimensional Euclidean space, see \cite{LSS1}. 

Of course, it is physically relevant and mathematically interesting to determine such a precise asymptotics also for ground states of fermions subject to an external field or even with interactions between them. From a rigorous point of view, the latter seems currently to be out of reach. Concerning external scalar fields there are publications devoted to free fermions in a (random) potential~\cite{EPS,MPS, MS, P,PasturSh} or in a one-dimensional periodic potential~\cite{PS}. As an aside, we mention that in the case of free fermions the large-scale behavior of the local entropy is not only known for the ground state, but also for the thermal equilibrium state at any temperature \cite{LSS2,LSS3}.

In the present paper, we (return to zero temperature and) consider the ground state of non-relativistic, spinless, and electrically charged fermions in the Euclidean plane $\R^2$ without interactions between them, but subject to an external magnetic field which is perpendicular to the plane and of constant strength $B>0$. This ground state became of interest in condensed-matter physics at first in the early 1930s for simplified explanations of the Landau diamagnetism and the De Haas--Van Alphen effect observed in metals, see \cite{Peierls, Huang}. The interest got revived and enhanced after the discovery of the (integrally) quantized Hall effect in certain quasi-two-dimensional semiconductor materials by Klaus von Klitzing in the year 1980, see~\cite{KvK, Huang}. To our knowledge, analytical contributions to the asymptotic growth of the local entropy of this ground state were made by Klich~\cite{Klich}, by Rodr{\'{\i}}guez and Sierra~\cite{RS1,RS2}, and recently by Charles and Estienne~\cite{CE}. All these authors consider the case of the lowest Landau level only. In addition, \cite{Klich, RS1, RS2} treat regions of simple geometric shape only. The important work of Rodr{\'{\i}}guez and Sierra~\cite{RS1} contains non-rigorous arguments, but their formula for the asymptotic coefficient $\mathsf M_0(h_{1})$ (see \eqref{def:entropy} and \eqref{coeff} for the definition), confirmed in \cite{CE}, has been a guide for us to arrive at the more general asymptotic coefficients presented here. After all, the simplicity of the coefficient $\mathsf M_\ell(f)$ of the sub-leading boundary-curve term in \eqref{RS}, see \eqref{hermite fct} and (\ref{mell:eq}), is striking, given that it results from the $2m$-fold integration in \eqref{term} for arbitrary exponent $m\ge 1$. 

By adapting Roccaforte's approach for translation invariant integral kernels \cite{R} to those of the Landau-projection operators we extend results in \cite{RS1,CE} to rather general regions, to an arbitrary single Landau-level eigenspace, and even to the orthogonal sum of the first $n+1$ eigenspaces for arbitrary $n\ge 0$. By the last extension we can allow for an arbitrary value of the chemical potential $\mu\ge B$  (in suitable physical units) and, hence, for an arbitrary areal density of the particle number. Our proof consists of two basic steps. In the first step, we present the precise asymptotics of the trace of smooth functions of localized (or spatially truncated) Landau projections. By a suitable application of the Stone--Weierstra\ss{} approximation theorem this asymptotics is shown to follow from a corresponding one for polynomials, based on Lemma \ref{lem:moments}. Our proof of Lemma \ref{lem:moments} is elementary in the sense that it does not make use of the quasi-classical functional calculus for pseudo-differential operators, not even of standard stationary-phase evaluation techniques. However, it involves one change of variables which is cumbersome to utilize, see \eqref{exponent}. In the second step, we show how to get from smooth functions to the R\'enyi entropy functions $h_\a$. This is not obvious, in particular for R\'enyi index $\a\le 1$, and does not follow from standard approximation schemes. Therefore we prove and employ certain Schatten--von Neumann quasi-norm estimates, similarly to what has been done in\cite{LSS1}.

Our main result is Theorem \ref{entropy of ground state}. It turns out that all local R\'enyi ground-state entropies grow to leading order proportional to $L$ when the dimensionless parameter $L>0$ of the scaled region $L\Lambda$ is sent to infinity (due to the off-diagonal Gaussian decay of the Landau-projection integral kernels). This is in agreement with the so-called area-law scaling \cite{ECP}. Given that, the corresponding proportionality factor has the form $\sqrt{B}|\p\L|\mathsf M_{\le \nu}(h_\a)$. Here, the first two factors are expected from considering physical dimensions, because $|\p\L|$ denotes the arc length of the (smooth) boundary curve $\p\L$. The third factor $\mathsf M_{\le \nu}(h_\a)$ is a dimensionless asymptotic coefficient precisely given by \eqref{mleell:eq} and \eqref{def:entropy}. It depends in a non-trivial way on the R\'enyi index $\a>0$, but on $B$ and $\mu$ only through the integer part $\nu$ of $(\mu/B-1)/2$. It is finite and positive, but in general a rather complicated expression. However, if $\mu<3B$, then it simplifies considerably, because only the lowest Landau-level eigenspace remains to be relevant. The result agrees (for $\alpha=1$) with the one proved recently by Charles and Estienne \cite{CE}.

\textbf{Acknowledgement:} Various parts of this paper were developed and written during several visits of HL and AVS to the FernUniversit\"at in Hagen in 2016--2019, and during the stay of all authors at the International Newton Institute (Cambridge, UK) in 2015 and to the American Institute of Mathematics (AIM) at San Jose (CA, USA) in 2017. The authors are grateful to these institutions for hospitality and financial support. AVS was also supported by the EPSRC grants EP/J016829/1 and EP/P024793/1. We thank Paul Pfeiffer for helpful remarks.

\section{Setting the stage and basic asymptotic results for smooth functions}

We denote the scalar products in the Euclidean plane $\R^2$ and in the Hilbert space $\text L^2(\mathsf D)$ of complex-valued, square-integrable functions on a Borel set $\mathsf{D}\subseteq \R^2$ by the same bracket $\langle\cdot|\cdot\rangle$ and use the same notation $\|\cdot\|$ for the induced norms. Our convention is that the scalar product is anti-linear in the first and linear in the second argument.

Since in the ideal Fermi gas the indistinguishable (point) particles do not interact with each other, it is sufficient to consider the common Schr\"odinger operator for the kinetic energy of a single particle in the plane subject to a perpendicular constant magnetic field of strength $B>0$. This operator is known as the \textit{Landau Hamiltonian}. It acts self-adjointly  on a dense domain of definition in the single-particle Hilbert space $\text L^2(\R^2)$  and is given by 
\be\label{landau}
\text H := (-\mathrm{i}\nabla - a)^2 \,.
\ee
Here, we choose the symmetric gauge $a(x)=(a_1(x), a_2(x)) := (x_2,-x_1)B/2$ for the vector potential $a: \R^2 \to\R^2$ generating the constant magnetic field (vector) perpendicular to the plane with Cartesian coordinates $x=(x_1, x_2)$. Other gauges yield operators being unitarily equivalent to $\text H$. Moreover, here and in the following we are using physical units such that the particle mass and the particle charge equals $1/2$ and $1$, respectively. Similarly, we put the speed of light, Planck's constant (divided by $2\pi$), and Boltzmann's constant all equal to $1$.

Throughout the paper we use the symplectic $2\times2$ matrix $\mathsf{J} := \begin{bmatrix}0&1\\-1&0\end{bmatrix}$, the generalized Laguerre polynomials
\[ \mathcal{L}^{(k)}_\ell(t) := \sum_{j=0}^\ell \frac{(-1)^j}{j!}\, \binom{\ell+k}{\ell - j}\, t^j\,,\quad k\in\{-\ell,-\ell+1,\dots\}\,,\quad t\ge0
\]
of degree $\ell\in\N_0$, and the abbreviation $\mathcal{L}_\ell := \mathcal{L}^{(0)}_\ell$. For each degree $\ell$ we define an infinite-dimensional projection (operator) $\text P_\ell$ on $\text L^2(\R^2)$ by the Hermitian integral kernel
\be \label{Landau_kernel:eq}
p_\ell(x,y) := \frac{B}{2\pi} \,\exp(-B\|x-y\|^2/4) \, \mathcal{L}_\ell(B\|x-y\|^2/2)\,\exp(\mathrm{i}\mfr{B}2\langle x|\mathsf{J}y\rangle)\,,\quad x,y\in\mathbb R^2\,.
\ee
It is obviously $\mathsf C^\infty$-smooth and a Carleman kernel in the sense that it is square integrable with respect to $y\in\R^2$ for all $x\in\R^2$, and vice versa. Now the spectral decomposition of the Landau Hamiltonian $\text H$ may be written as 
\be \label{sr}
\text H=B\sum_{\ell=0}^\infty(2\ell +1) \text P_\ell \,.
\ee
As usual, this formula is meant in the sense of strong operator convergence on $\text L^2(\R^2)$. It goes back to Fock\cite{Fock} and Landau\cite{Landau}. The projections $\text P_\ell$, now recognizable as spectral projections, depend on the chosen gauge through the last (complex-valued phase) factor in \eqref{Landau_kernel:eq}, but the set $ \{ B, 3B, 5B,\ldots\}$ of harmonic-oscillator like eigenvalues, in other words \textit{Landau levels}, does not. The degree $\ell$ is now called \textit{Landau-level index}. We will also need the projection $\text P_{\le n} := \sum_{0\le \ell\le n} \text P_\ell$ on the orthogonal sum of the first  $n+1$ Landau-level eigenspaces $\text P_\ell\text L^2(\R^2)$ and mention the functional relation $\mathcal{L}_{\le n} := \sum_{0\le \ell\le n} \mathcal{L}_\ell = \mathcal{L}_n^{(1)}$. For later purposes we single out the (translation invariant) Gaussian part of the kernel by defining
\begin{align*}
g(z) := &\ \left(\frac{B}{2\pi}\right)^{1/2} \,\exp\bigg(-\frac{Bz^2}{4}\bigg)\,,\quad z\in\mathbb R\,,\\
g_2(x): = &\ g(x') g(x'')\,,\quad x = (x', x'')\in\mathbb R^2\,,
\end{align*}
so that 
\begin{align*}
\text P_{\ell}(x,y) \equiv & \,p_\ell(x,y)= \mathcal{L}_\ell(B\|x-y\|^2/2) g_2(x-y)\,\exp(\mathrm{i}\mfr{B}2\langle x|\mathsf{J}y\rangle)\,,\\
\text P_{\le n}(x,y):=& \sum_{\ell=0}^n p_\ell(x,y) = \mathcal{L}_{\le n}(B\|x-y\|^2/2) g_2(x-y)\,\exp(\mathrm{i}\mfr{B}2\langle x|\mathsf{J}y\rangle)\,.
\end{align*}

Now we are prepared to turn to the ground state of the ideal Fermi gas with the Landau Hamiltonian, see \eqref{landau} and \eqref{sr}, as its single-particle Hamiltonian and with the chemical potential $\mu\ge B$ as a real parameter. According to the grand-canonical formalism of quantum statistical mechanics \cite{Huang,BR2} this (infinite-area) ground state is quasi-Gaussian (in other words, quasi-free) and, as such, characterized by its reduced single-particle density operator on $\text L^2(\R^2)$ given by the \textit {Fermi projection}
\be\label{gs}
\Theta(\mu \mathds{1}-\text H)=\sum_{\ell=0}^\infty\Theta\big(\mu-(2\ell +1)B\big) \text P_\ell = \text P_{\le \nu} \,,\quad \mu\ge B \,.
\ee
Here, $\Theta$ is Heaviside's unit-step function (defined by $\Theta(t):=1$ if $t\ge0$ and zero otherwise), $\mathds{1}$ denotes the identity operator on $\text{L}^2(\mathbb R^2)$, and  $\nu :=\lfloor{(\mu/B-1)/2}\rfloor$ is the integer part of $(\mu/B-1)/2\ge 0$. 
Now we consider a Borel set  $\L\subseteq \R^2$  and the multiplication operator $\mathds{1}_\L$ on $\text{L}^2(\mathbb R^2)$ corresponding to its indicator function $1_\L$ on
 $\R^2$. Moreover, we introduce the \textit{local(ized)} Landau projections
 
\begin{align*}
\text P_\ell(\L) : = \mathds{1}_\L \text P_\ell \mathds{1}_\L\,,\ 
\text P_{\le n}(\L) : = \mathds{1}_\L \text P_{\le n} \mathds{1}_\L\,.
\end{align*}

The quasi-Gaussian substate associated with $\L\subseteq\R^2$ is now simply characterized by the \textit{local(ized)} Fermi projection
\be\label{ss}
\mathds{1}_\L \Theta(\mu \mathds{1}- \text H )\mathds{1}_\L =\text P_{\le \nu}(\L)\, ,
\ee 
see \cite {HLS}.

We ignore the uninteresting case $\mu<B$, because then the Fermi projection is the zero operator corresponding to a vanishing number of particles. In contrast, as a function of $\mu\ge B$ the (mean) local areal density of the particle number, $\rho(x)$, in the ground state characterized by the density operator \eqref{ss} is non-zero and equal to the diagonal of its integral kernel, that is,

\be\label{density}
\rho (x)=\text P_{\le \nu}(\L)(x,x)=   (\nu +1) \frac{B}{2\pi}1_\L(x)\,,\quad x\in\R^2 \,.
\ee
Integration over the plane $\R^2$ gives the (mean) total number of particles, $(\nu +1)B|\L| /2\pi$, in its subset $\L$ with (Lebesgue) area $|\L|$. So $\nu +1$ corresponds to an integer value of the filling factor in the physics literature.
In view of \eqref{ss} it suffices in the following to consider the projection $\text P_{\le n}(\L)$ for arbitrary $n\in\N_0$. Moreover, from now on we typically assume that $\Lambda\subset\R^2$ is the union of finitely many {bounded} domains (open connected sets), such that their closures are pairwise disjoint. We call such a $\Lambda$ a \textit{bounded region}. If the boundary curve $\partial\Lambda$ of $\L$ is $\mathsf C^{\g}$, $\g \in\N\cup\{\infty\}$, then we say that $\L$ is a \emph{bounded $\mathsf C^{\g}$-region}. 

Before we state our basic asymptotic results we recall the definition of the Hermite polynomials, $H_\ell$, of degree $\ell\in\N_0$. They satisfy the orthogonality relation
\be\label{orthog} \int_\R \text d t\, \exp(-t^2) H_\ell(t) H_{\ell'}(t) = \sqrt{\pi} 2^\ell \ell! \,\delta_{\ell,\ell'}\,,\quad \ell,\ell'\in\N_0\,. 
\ee
An explicit formula is
\be \label{Hermite explicit}
H_\ell(t) = \ell! \sum_{j=0}^{\lfloor \ell/2\rfloor} \frac{(-1)^j}{j!(\ell-2j)!} (2t)^{\ell-2j}\,,\quad t\in\R, \ell\in\N_0\,.
\ee
The Hermite functions $\psi_\ell$, defined by
\be\label{hermite fct}
\psi_\ell(t) := (\sqrt{\pi} 2^\ell \ell!)^{-1/2} H_\ell(t) \exp(-t^2/2)\,,\quad t\in\R, \ell\in\N_0\,,
\ee
constitute an orthonormal basis of the Hilbert space $\text L^2(\R)$ and are the (energy) eigenfunctions of the one-dimensional harmonic oscillator, that is,
\[ -\psi_\ell''(t) + t^2\psi_\ell(t) = (2\ell+1) \psi_\ell(t)\,,\quad t\in\R\,.
\]
For $\xi\in\R$ and a complex-valued function $f$ on the closed unit interval ${[0,1]}$ as in Lemma \ref{lemma 3} below we define 
\begin{align}\label{mell:eq}
\l_\ell(\xi):= \int_\xi^\infty {\mathrm d} t\, \psi_\ell(t)^2\,,\ \quad
{\sf M}_\ell(f) := \int_\R \frac{{\mathrm d} \xi}{2\pi}\,[f(\l_\ell(\xi)) - f(1)\l_\ell(\xi)]\, , 
\quad \ell\in\mathbb N_0.
\end{align} 
Obviously, each function $\l_\ell$ takes values in ${[0,1]}$ and is (strictly) decreasing. We also need to introduce for each $n\in\N_0 $ the one-parameter family of operators
\be\label{opK}
\mathcal K_{n, \xi} := \sum_{\ell=0}^n |\psi_{\ell, \xi}\rangle\langle\psi_{\ell, \xi}| =  \mathds{1}_{[\xi,\infty)}\sum_{\ell=0}^n |\psi_{\ell}\rangle\langle\psi_{\ell}|\, \mathds{1}_{[\xi,\infty)}    \,,\quad  \xi\in\R\,.
\ee
The operator $\mathcal K_{n, \xi}$  maps ${\mathrm L}^2(\mathbb R)$ 
self-adjointly on its $(n+1)$-dimensional subspace spanned by the first $(n+1)$ truncated 
Hermite functions $\psi_{\ell, \xi} := \psi_\ell 1_{[\xi, \infty)}$. 
This operator is not a projection, 
but it satisfies $0\le \mathcal K_{n,\xi} \le\mathds{1}_{[\xi,\infty)}\le \mathds{1}$.
Its integral kernel is given by the sum $\sum_{0\le \ell\le n} \psi_{\ell, \xi}(t)\psi_{\ell, \xi}(t')$,  which can be evaluated explicitly, see \eqref{op: K} below.

Along with ${\sf M}_\ell(f)$ we also define for $n\in\N_0$
\begin{align}\label{mleell:eq}
{\sf M}_{\le n}(f)
:= \int_\R \frac{{\mathrm d} \xi}{2\pi}\,[\tr f(\mathcal K_{n,\xi}) - f(1)\,\tr \mathcal K_{n,\xi}]\,,\quad {\sf M}_{\le 0}(f) = {\sf M}_{0}(f)\,.
\end{align} 
Here the trace refers to operators on ${\mathrm L^2}(\mathbb R)$. Now we are in a position to present our two basic asymptotic results. 

\begin{thm}[\bf{For the $\ell${th} Landau level}, 
$\ell\in\N_0$]\label{thm:LL smooth f} 
Let $\L\subset\R^2$ be a bounded $\mathsf C^3$-region in the sense defined below \eqref{density}. 
Moreover, let $f:[0, 1]\to \mathbb C$ be a complex-valued continuous 
function on the closed unit interval with $f(0)=0$, differentiable from the right 
at $t=0$ and differentiable from the left 
at $t=1$. Finally, let $L>0$ be a (dimensionless) scaling parameter. Then we have
\be \label{RS}
\tr f({\rm P}_\ell(L\L)) = L^2B\, \frac{|\L|}{2\pi}\,f(1)  + L\sqrt{B} \, |\p\L|\,{\sf M}_\ell(f) + o(L)\,,
\ee
as $L\to\infty$. The asymptotic coefficient is finite, that is, $|{\sf M}_\ell(f)|<\infty$.
\end{thm}
Here the trace refers to operators on ${\mathrm L}^2(\mathbb R^2)$ and, as usual, $o(L)$ stands for some function of $L$ with $\lim_{L\to\infty}|o(L)|/L=0$.

\begin{thm}[\bf{For the first $(n+1)$ Landau levels}, $n\in\N_0$]\label{thm:ground state smooth f} 
Under the same assumptions as in Theorem \ref{thm:LL smooth f} we have
\be \label{RS:general}
\tr f({\rm P}_{\le n}(L\L)) = L^2B\,\frac{|\L|}{2\pi}\,(n+1) f(1) + L\sqrt{B} \, |\p\L|\, {\sf M}_{\le n}(f) + o(L)\,,
\ee
as $L\to\infty$. The asymptotic coefficient is finite, that is, $|{\sf M}_{\le n}(f)|<\infty$.
\end{thm}

The finiteness of the coefficients  ${\sf M}_\ell(f)$ and ${\sf M}_{\le n}(f)$ are consequences of the following Lemma \ref{lemma 3} and Lemma \ref{mleell:lem}, because the smooth function $f$ assumed in Theorem \ref{thm:LL smooth f} and Theorem \ref{thm:ground state smooth f} satisfies the bound \eqref{q:eq}. The proofs of \eqref{RS} and \eqref{RS:general} are postponed until the proof of Lemma \ref{mleell:lem}. In the next lemma and in the following, by $C, c$ with or without indices, we denote various finite and positive constants, whose precise values are of no importance.

\begin{lemma}\label{lemma 3}
Let $f:[0, 1]\to \mathbb C$ be a measurable function satisfying the bound 
\begin{align}\label{q:eq}
|f(t) - f(1) t|\le Ct^q(1-t)^q, \quad t\in [0, 1]\,,
\end{align}
with some $q>0$. Then $|{\sf M}_\ell(f)|<\infty$ for all $\ell \in\N_0$.
\end{lemma} 

\begin{proof}
Firstly, we observe that 
\begin{align*}
|\psi_\ell(t)|\le C(1+ |t|)^\ell \,\e^{-\frac{t^2}{2}},\quad t\in \mathbb R, \ell\in \mathbb N_0\,,
\end{align*}
with a constant $C$ depending on $\ell$. Therefore, for $\xi \ge 0$ the function $\l_\ell$ satisfies the bound 
\begin{align*}
\l_\ell(\xi)\le C \int_\xi^\infty \text dt\, (1+t)^{2\ell} \,\e^{-t^2} \le C_\d \,\e^{-\d \xi^2}\,,
\end{align*}
with an arbitrary $\d < 1$. Similarly, for $\xi <0$ we have
\begin{align*}
\l_\ell(\xi) = 1- \l_\ell(-\xi)\ge 1-C_\d \,\e^{-\d \xi^2}\,.
\end{align*}
Combining this with \eqref{q:eq} yields the claimed result.
\end{proof}

Concerning the other coefficient ${\sf M}_{\le n}(f)$ we have the following

\begin{lemma}\label{mleell:lem}
Under the same assumption as in Lemma \ref{lemma 3} we have 
\begin{align*}
\|f(\mathcal K_{n, \xi}) - f(1) \mathcal K_{n, \xi}\|_1\le C_\d \,\e^{-\d q \xi^2}\,,
\end{align*}
for every $n\in\N_{0}$ with an arbitrary $0<\d<1$, and hence $|{\sf M}_{\le n}(f)|<\infty$.
\end{lemma}

Before proving this lemma we compile, for the reader's convenience, some basic properties of the Schatten--von Neumann classes, $\mathfrak S_p, 0<p<\infty$, of compact operators, see \cite{BS_book,Si}. By $s_n({\rm T})$ with $n\in\N$, we denote the singular values of a compact operator ${\rm T}$ on an abstract (separable) Hilbert space, enumerated in decreasing order. Then, the operator ${\rm T}$ is said to belong to $\mathfrak S_p$ if it has the finite \textit{Schatten--von Neumann (quasi-)norm}
\begin{align*}
\|{\rm T}\|_p := \bigg[\sum_{n=1}^\infty s_n({\rm T})^p\bigg]^{\frac{1}{p}}<\infty\,.
\end{align*}
If $p\ge 1$, then $\| \cdot \|_p$ is a norm. If $0 < p < 1$, then it is a quasi-norm which satisfies the 
\textit{$p$-triangle inequality}
\begin{align}\label{p-tri:eq}
\| {\rm T}_1+{\rm T}_2\|_p^p\le \|{\rm T}_1\|_p^p + \|{\rm T}_2\|_p^p\,.
\end{align}
The class $\mathfrak S_1$ is the standard trace class. For ${\rm T}\in\mathfrak S_1$ its trace, $\tr {\rm T}$, is well-defined and satisfies $|\tr {\rm T}|\leq \|{\rm T}\|_1 $. If ${\rm T}\ge 0$, then $\tr {\rm T}=\|{\rm T}\|_1 $. We also note that the usual (uniform) operator norm $\|\cdot\| $ may be viewed as $\| \cdot \|_p$ in the limit $p\to\infty$. Finally, we mention that $\|\cdot \|_p$ satisfies a H\"older-type inequality in the sense that
\begin{equation}\label{holder}
\|{\rm T}_1 {\rm T}_2 \|_p\leq\|{\rm T}_1\|_{p_{1}}\|{\rm T}_2\|_{p_{2}}
\end{equation}
for any $p_1, p_2 \in (0,\infty]$ with $ 1/p_1 + 1/p_2=1/p$.
\begin{proof}[Proof of Lemma \ref{mleell:lem}]
Let us now consider the operator $\mathcal K_{n, \xi} = \sum_{\ell=0}^n Q_{\ell, \xi}$ defined in \eqref{opK}, where we have put $Q_{\ell, \xi}:= |\psi_{\ell, \xi}\rangle\langle \psi_{\ell, \xi}| $ With
$\tilde{g}(t):= t(1-t)$ we then have
\begin{align*}
\tilde{g}\big(\mathcal K_{n, \xi}\big) 
= \sum_{\ell=0}^n (Q_{\ell, \xi} - Q_{\ell, \xi}^2) 
- \sum_{\ell,\ell'=0, \ell\not = \ell'}^n Q_{\ell, \xi} Q_{\ell', \xi}\,.
\end{align*}
Since each operator $Q_{\ell, \xi}$ is one-dimensional, we easily find that
\begin{align*}
\|Q_{\ell, \xi} - Q_{\ell, \xi}^2\| = \l_\ell(\xi)\big(1-\l_\ell(\xi)\big)\le C_\d \,\e^{-\d \xi^2},\ \xi\in\mathbb R\,,
\end{align*}
according to the proof of Lemma \ref{lemma 3}. Furthermore,
\begin{align*}
\| Q_{\ell, \xi} Q_{\ell', \xi}\| = |\langle \psi_{\ell, \xi}| \psi_{\ell', \xi}\rangle| \,\|\psi_{\ell, \xi}\|\,\|\psi_{\ell', \xi}\|\le |\langle \psi_{\ell, \xi}| \psi_{\ell', \xi}\rangle|\,.
\end{align*}
Consequently, for $\xi \ge 0$, we have 
\begin{align*}
\| Q_{\ell, \xi} Q_{\ell', \xi}\|\le \|\psi_{\ell, \xi}\| \|\psi_{\ell', \xi}\|\le \sqrt{\l_\ell(\xi)\l_{\ell'}(\xi)}\le C_\d \,\e^{-\d\xi^2}\,. 
\end{align*}
For the case $\xi <0$ we observe that $\langle\psi_\ell|\psi_{\ell'}\rangle = 0$, $\ell\not = \ell'$, so that  
$\langle \psi_{\ell, \xi}| \psi_{\ell', \xi}\rangle = - \langle \psi_\ell - \psi_{\ell, \xi}| \psi_{\ell'} - \psi_{\ell', \xi}\rangle$, and hence 
\begin{align*}
\| Q_{\ell, \xi} Q_{\ell', \xi}\| \le \sqrt{\l_\ell(-\xi)\l_{\ell'}(-\xi)}\le C_\d \,\e^{-\d \xi^2}\,.
\end{align*}
Collecting the above bounds we conclude that 
\begin{align}\label{knorm:eq}
\|\tilde{g}(\mathcal K_{n, \xi})\|\le C_\d  \,\e^{-\d \xi^2}\,.
\end{align} 
Using now \eqref{q:eq} we have
\begin{align*}
\|f(\mathcal K_{n, \xi}) - f(1) \mathcal K_{n, \xi}\|_1\le C\|\tilde{g}(\mathcal K_{n, \xi})\|_q^q\,.
\end{align*}
Since the operator $\mathcal K_{n, \xi}$ has finite dimension $n+1$, the right-hand side of the last inequality
is bounded from above by $C\|\tilde{g}(\mathcal K_{n, \xi})\|^q$. Therefore \eqref{knorm:eq} leads to the claimed result. 
\end{proof}

\begin{proof}[Proofs of Theorem \ref{thm:LL smooth f} and Theorem \ref{thm:ground state smooth f}]
By linearity, Lemma \ref{lem:moments} and Lemma \ref{lem:general moments} in the next section imply Theorem \ref{thm:LL smooth f} and \ref{thm:ground state smooth f}, respectively, for an arbitrary polynomial $f$ (with $f(0)=0$). So here we only need to show how to extend the claimed results \eqref{RS} and \eqref{RS:general} from polynomials to the smooth function $f$ assumed in Theorem \ref{thm:LL smooth f} and Theorem \ref{thm:ground state smooth f}. This is by now standard and several versions of this extension are available, e.g. \cite{PS, R, Sob:AMS, Widom}. Here we follow the recent one in \cite{PS}. As a by-product we get the a-priori finiteness of the left-hand side of \eqref{RS}, see \eqref{sup} and \eqref{inf}. The a-priori finiteness of the left-hand side of \eqref{RS:general} follows similarly.

Without loss of generality we may assume that $f$ is real-valued. Besides the necessary condition $f(0)=0$ we may assume that $f(1)=0$. This can be achieved by replacing $f(t)$ with $f(t) - f(1)t$. The function $f$ has the form $f= b\tilde{g}$ with $\tilde{g}(t) = t(1-t)$, from above, and with some real-valued continuous function $b$ on $[0,1]$. According to the Stone--Weierstra\ss{} approximation theorem, there exists for any given $\varepsilon>0$ a real-valued polynomial $p$ on $[0,1]$ such that $\sup_{t\in[0,1]} |p(t)-b(t)|\le \varepsilon$. Thus with $\tilde{p}:=\tilde{g} p$ we have 
\begin{align*}
\tilde{p}(t) - \varepsilon \tilde{g}(t) \le f(t) \le \tilde{p}(t) + \varepsilon \tilde{g}(t)\,, \quad t\in [0, 1]\,,
\end{align*} 
and hence
\begin{align*} 
\tr \tilde{p}(\text P_\ell(L\L)) - \varepsilon \,\tr \tilde{g}(\text P_\ell(L\L)) &\le 
\tr f(\text P_\ell(L\L)) \le  \tr \tilde{p}(\text P_\ell(L\L)) + \varepsilon \,\tr \tilde{g}(\text P_\ell(L\L))\,.
\end{align*}
Using \eqref{RS} for $\tilde p$ and $\tilde{g}$, we arrive at the bound 
\begin{align}\label{sup}
\limsup_{L\to\infty} \frac{\tr f(\text P_\ell(L\L))}{L\sqrt{B}|\p\L|} 
&\le {\sf M}_\ell(\tilde p) + \varepsilon {\sf M}_\ell(\tilde{g})\,.  
\end{align}
Since ${\sf M}_\ell(\tilde p) = {\sf M}_\ell(f) + {\sf M}_\ell(\tilde{g}(p-b))
\le {\sf M}_\ell(f) + \varepsilon{\sf M}_\ell(\tilde{g})$, the right-hand side of the above inequality does not exceed 
$ {\sf M}_\ell(f) + 2\varepsilon{\sf M}_\ell(\tilde{g})$.
Similarly,
\begin{align}\label{inf} 
\liminf_{L\to\infty} 
\frac{\tr f(\text P_\ell(L\L))}{L\sqrt{B}|\p\L|} \ge 
{\sf M}_\ell(f) - 2\varepsilon{\sf M}_\ell(\tilde{g})\,.
\end{align}
Since $\varepsilon$ is arbitrary, formula \eqref{RS} follows. 

In order to prove Theorem \ref{thm:ground state smooth f} we use the same argument as above for the operator $\text P_{\le n}(L\L)$ and the asymptotic coefficient ${\sf M}_{\le n}(f)$. 
\end{proof}

\section{Underlying asymptotic results for polynomials}

We have seen that the asymptotic results of Theorem \ref{thm:LL smooth f} and Theorem \ref{thm:ground state smooth f} rely on corresponding results for polynomials $f$. By linearity, it suffices to consider natural powers of the corresponding projections. We begin with 

\begin{lemma} \label{lem:moments} 
Let $\L\subset\R^2$ be a bounded $\mathsf C^3$-region and $m\in\N$. Then we have for any $\ell\in\N_0$
\be \label{13}
\tr{ \rm P}_\ell(L\L)^m =  L^2 B \frac{|\L|}{2\pi}\,  + L\sqrt{B}|\p\L| \, \int_\R \frac{\mathrm{d}\xi}{2\pi}\,\big[ \l_\ell(\xi)^m-\l_\ell(\xi)\big] + \mathcal{O}(1)\,,
\ee
as $L\to\infty$.
\end{lemma}
Here, as usual, $\mathcal{O}(1)\equiv \mathcal{O}(L^0)$ stands for some function of $L$ with $\limsup_{L\rightarrow\infty}|\mathcal{O}(1)|<\infty$.

\begin{proof} 

At first we note that $\mathds{1}_\L \text P_\ell$  is a Hilbert--Schmidt operator on $\text L^2(\R^2)$, equivalently $\| \mathds{1}_\L \text P_\ell\|_2<\infty$, because its integral kernel $1_\L(x) p_\ell(x,y)$ is square-integrable. Therefore $\text P_\ell(\L) = (\mathds{1}_\L \text P_\ell) (\text P_\ell\mathds{1}_\L)$ is a trace-class operator and so are its natural powers, due to $ \|\text P_\ell(\L)\| \le 1$. The trace of the power $\text P_\ell(\L)^m$ can be calculated by integrating the diagonal of its integral kernel, that is,
 
\be\label{diagonal}
\tr \text P_\ell(\L)^m = \int_{\R^2} \text dx\,\text P_\ell(\L)^m(x,x) \,.
\ee
This follows from the continuity of $\text P_\ell(\L)^m (x,y)$ as a function of $(x,y) \in \L\times\L$ which, in turn, follows from the m-fold iteration of $\text P_\ell(\L)(x,y)= p_\ell(x,y) 1_\L(x)1_\L(y) $ and the dominated-convergence theorem. If $x\not\in\L$ or $y\not\in\L$, then $\text P_\ell(\L)^m (x,y)=0$. We proceed with \eqref{diagonal}. Since \eqref{13} is now seen to be true for $m=1$, we assume from now on $m\ge2$. Then the diagonal $\text P_\ell(\L)^m(x,x)$ is given by
\be\label{itkernel} 
1_\L(x)\,\int_{\R^{2(m-1)}} \text dx_1\cdots \text dx_{m-1}\, p_\ell(x,x_1) p_\ell(x_1,x_2)\cdots p_\ell(x_{m-2},x_{m-1})p_\ell(x_{m-1},x)\, 1_\L(x_1) \cdots 1_\L(x_{m-1}) \,.
\ee
It is convenient to change to new integration variables $\mathbf y := (y_1,\ldots,y_{m-1})$ defined by $y_1:=x-x_1,y_2:=x_1-x_2,\ldots, y_{m-1}:=x_{m-2}-x_{m-1}$. Furthermore, we set $y_m := y_1+\cdots+y_{m-1}$. Then, $x_1 = x-y_1,x_2 = x-y_1-y_2,\ldots,x_{m-1} = x-y_m$ and
\beax \langle x|\mathsf{J} x_1\rangle &=& -\,\langle x|\mathsf{J} y_1\rangle 
\\
\langle x_1|\mathsf{J} x_2\rangle &=& -\,\langle x-y_1|\mathsf{J} y_2\rangle
\\
\langle x_2|\mathsf{J} x_3\rangle &=& -\,\langle x-y_1-y_2|\mathsf{J} y_3\rangle
\\
\vdots& &\vdots
\\
\langle x_{m-2}|\mathsf{J} x_{m-1}\rangle &=& -\,\langle x-y_1-\cdots-y_{m-2}|\mathsf{J} y_{m-1}\rangle
\\
\langle x_{m-1}|\mathsf{J} x\rangle &=& -\,\langle y_m|\mathsf{J} x\rangle\,.
\eeax
With $x_0:=x_m:=x$ we therefore have
\be \label{equ 5}
\sum_{i=0}^{m-1} \langle x_i|\mathsf{J} x_{i+1}\rangle = \sum_{i=1}^{m-2}\langle \sum_{j=1}^i y_j|\mathsf{J}y_{i+1}\rangle\,.
\ee
If $m=2$ then the left-hand side of  \eqref{equ 5} is zero and its right-hand side is meant to be 0. 
By combining \eqref{diagonal}, \eqref{itkernel}, \eqref{Landau_kernel:eq}, and 
\eqref{equ 5} the trace of $\text P_\ell(\L)^m$ can now be written as
\be\label{moment}
\tr \text P_\ell(\L)^m =\int_{\R^{2(m-1)}} 
\text d\mathbf y\,\mathfrak f_m (\mathbf y)\int_{\R^2} \text dx\, 1_\L(x) 1_\L(x-y_1) \cdots 1_\L(x-y_{m})\,,
\ee
with the function $\mathfrak f_m$ defined by
\be\label{helpf}
\mathfrak f_m(\mathbf y) := \Big[\prod_{j=1}^{m} \mathcal{L}_\ell(B\|y_j\|^2/2) g_2(y_j)\Big] \, \exp\Big[\mathrm{i}\mfr{B}{2} \sum_{i=1}^{m-2}\langle \sum_{j=1}^i y_j|\mathsf{J}y_{i+1}\rangle \Big]\,.
\ee
Now we insert the scaling parameter $L>0$ and apply Roccaforte's asymptotic expansion of Appendix \ref{App:R} up to the first order
\begin{align}
\int_{\R^2} \text dx\, &1_{L\L}(x) 1_{L\L}(x-y_1) \cdots 1_{L\L}(x-y_m) \notag
\\
&=\big|L\L\cap(y_1+L\L)\cap(y_1+y_2+L\L)\cap\cdots\cap(y_{m}+L\L)\big|
\notag\\ 
&=|L\L| - \big|L\L \setminus \big(L\L\cap(y_1+L\L)\cap(y_1+y_2+L\L)\cap\cdots\cap(y_{m}+L\L)\big)\big|
\notag\\
&=\label{20} L^2|\L| - L \int_{\p\L} \text dA(x) \, \max\big\{0,\langle y_1|n_x\rangle,\langle y_1+y_2|n_x\rangle,\ldots, \langle y_{m} |n_x\rangle\big\}+W\bigg(\sum_{j=1}^m \|y_j\|^2 \bigg)\,,
\end{align}
as $L\to\infty$. Here, $A$ is the canonical arc-length measure on $\p\L$, $n_x$ is the inward unit normal vector at the point $x\in\partial\L$, and $W$ is a function such that $W(s)\le C s$, $s\ge0$. We scale $\mathbf y$ by $B^{-1/2}$ and $\L$ by $B^{1/2}$. Then we can set from now on $B=1$ in the function $\mathfrak f_m$ and replace $L\L$ by $\sqrt{L^2 B}\L$. The parameter that tends to infinity in our asymptotic analysis is thus effectively $\sqrt{L^2B}$. 

For a given point on the boundary curve, $x\in\p\L$, we decompose each vector $y_i\in\R^2$ into a component parallel and a component perpendicular to the tangent (line) $\text T_x(\partial\L)\cong \R$ at $x\in\p\L$
according to
\be \label{zt coordinates}
y_i = -z_i \mathsf{J} n_x +  t_i n_x\,,\quad i=1,\ldots,m-1\,,
\ee 
with the real numbers $t_i:=\langle y_i|n_x\rangle$ and $z_i:= -\langle y_i|\mathsf{J} n_x \rangle$ so that $\|y_i\|^2 = z_i^2 + t_i^2$. Then we get
\be\label{zt1}
\mathcal{L}_\ell(\|y_i\|^2/2)g_2(y_i) = \mathcal{L}_\ell((z_i^2 + t_i^2)/2)g(z_i) g(t_i)
\ee
and
\be\label{zt2} 
\sum_{i=1}^{m-2}\langle \sum_{j=1}^i y_j|\mathsf{J}y_{i+1}\rangle = \sum_{i=1}^{m-1} z_i \sum_{j=1}^{m-1} \mathsf{S}_{ij} t_j = \langle \mathbf z|\mathsf{S} \mathbf t\rangle\, ,
\ee
to be used on the right-hand side of \eqref{helpf}. Here, $\mathbf z := (z_1,\ldots,z_{m-1})$ and $\mathbf t := (t_1,\ldots,t_{m-1})$. Moreover, $\mathsf{S}$ is the $(m-1)\times(m-1)$ matrix with entries
\be \mathsf{S}_{ij} := \left\{\begin{array}{cll}-1&\mbox{ if }&i<j\\0&\mbox{ if }&i=j\\1&\mbox{ if }&i>j\end{array}\right.\,.
\ee
By setting $t_0:=0$ the maximum in \eqref{20} can now be written as follows
\be\label{zt3}
\max\big\{0,\langle y_1|n_x\rangle,\langle y_1+y_2|n_x\rangle,\ldots, \langle y_1+\cdots+ y_{m-1}|n_x\rangle\big\}= \max_{0\leq q\leq  m-1}\sum_{r=0}^{q} t_r =: M(\mathbf t)\,.
\ee
Let us now introduce new variables $(T_1,\ldots,T_{m-1})$ by the sums
\be \label{def T}
T_i:= \sum_{j=1}^{m-1} \mathsf{S}_{ij} t_j\,.
\ee
We also define $T_m:=0$, $t_m:=t_1+\cdots + t_{m-1}$, and $z_m := z_1+\cdots+z_{m-1}$.

The change \eqref{zt coordinates} from the (global) variables $y_i$ to the $x$-dependent (local) variables $(z_i,t_i)$ corresponds to a translation and a rotation of the coordinate system. This implies that $\text dy_i = \text dt_i \text dz_i$ which is shorthand for the underlying invariance of the multi-dimensional Lebesgue measure. Once the integration with respect to all the variables $z_i$ and $ t_i$ is done, the result will turn out to be independent of $x\in\p\L$ and the remaining integration with respect to $x$ along the boundary curve $\p\L$ simply yields the factor $L\sqrt{B}|\p\L|$.

By combining \eqref{moment}, \eqref{helpf}, \eqref{20}, \eqref{zt1}, \eqref{zt2}, and \eqref{zt3}, and by referring to the Fubini--Tonelli theorem we get for the time being
\be\label{moment2}
\tr \text P_\ell(\L)^m= \int_{\R^{m-1}} \text d\mathbf t \,g(t_1)\cdots g(t_{m-1}) g(t_m) I_m(\mathbf t)\Big(L^2 B |\L| -L\sqrt{B}\int_{\p\L} \text dA(x)\,M(\mathbf t)\Big) + \mathcal{O}(1)\,,
\ee
with
\be \label{def I}
I_m(\mathbf t) := \int_{\R^{m-1}} \text d\mathbf z\, \prod_{j=1}^{m} g_{T_j,t_j}(z_j) 
\ee
and
\be 
g_{T,t}(z) := \mathcal{L}_\ell((z^2+t^2)/2) \,g(z) \, \exp(\mfr{\mathrm{i}}{2}Tz)\,,\quad     T, t, z \in \R\,.
\ee
[When it comes to integration we do not switch from the $t$-variables to the $T$-variables; moreover, we note that $\det\mathsf{S}=0$, resp. $=1$ if $m$ is even, resp. odd.]
The integral $I_m(\mathbf t)$ can be viewed as the $m$-fold convolution product $g_{T_1,t_1}\ast\cdots\ast g_{T_m,t_m}$ evaluated at $0$. This suggests to introduce the (inverse) Fourier transform
\begin{align}&\widecheck{g}_{T,t}(\xi) := \frac1{\sqrt{2\pi}} \,\int_\R \text d\omega\, g_{T,t}(\omega) \, \exp(\mathrm{i}\omega\xi) 
\notag\\
&=\frac{1}{2\pi}\, \int_\R \text d\omega\, \mathcal{L}_\ell((\omega^2+t^2)/2)\exp\big(-\omega^2/4 +\mathrm{i} T\omega/2 + \mathrm{i} \omega\xi\big)
\\
&=\frac{1}{2\pi}\,\exp[-(\xi+T/2)^2] \, \int_\R \text d\omega\, \mathcal{L}_\ell\big((\omega+\mathrm i(T+2\xi))^2+t^2)/2\big)\,\exp(- \omega^2/4)\,.
\end{align}
If $\ell=0$, then this integral can be calculated explicitly. But even then it turns out to be more convenient 
not to perform this integration at this point.

Therefore, the $(m-1)$-fold integral \eqref{def I} can be rewritten as an integral over the real line according to
\be \label{U-turn} 
I_m(\mathbf t) = (2\pi)^{m/2} \,\int_\R \frac{\text d\xi}{2\pi} \, \prod_{j=1}^m \widecheck{g}_{T_j,t_j}(\xi)\,,
\ee
and the term of the sub-leading order $L$ in \eqref{moment2} becomes equal to (using the notation $\boldsymbol{\omega}:=(\omega_1,\ldots,\omega_m)$)
\begin{align}\label{term}
-&L\sqrt{B}|\p\L| \int_{\R^{m-1}} \text d\mathbf t\, M(\mathbf t)g(t_1)\cdots g(t_{m-1}) g(t_m) \, I_m(\mathbf t)\notag
\\
&=-L\sqrt{B}|\p\L|(2\pi)^{-m/2} \int_{\R^{m}} \text d\boldsymbol{\omega} \int_\R\frac{\text d\xi}{2\pi} \int_{\R^{m-1}} \text d\mathbf t\, M(\mathbf t)\notag
\\
&\phantom{-L\sqrt{B}}\times\prod_{j=1}^m g(t_j) \exp\big(-(\xi+T_j/2)^2\big)  \mathcal{L}_\ell\big((\omega_j+\mathrm i(T_j+2\xi))^2+t_j^2)/2\big) \exp(-\omega_j^2/4)\,.
\end{align}
Following Roccaforte \cite{R} we now introduce $m(\geq 2)$ subsets $ \mathcal S, \mathcal S_1,\ldots,\mathcal S_{(m-1)}$ of $\R^{m-1}$ by
\be \mathcal S:=\big\{\mathbf t\in\R^{m-1} : M(\mathbf t)>0\big\}
\ee
and
\be \mathcal S_q := \Big\{\mathbf t\in\R^{m-1} : \sum_{r=s}^{q} t_r >0 \mbox{ for } 1\le s\le q \mbox{ and } \sum_{r=q+1}^{q+p} t_r<0 \mbox{ for } 1\le p\le m-1-q\Big\}\,
\ee
for $1\le q\le m-1$.
The sets $\mathcal S_q$ are pairwise disjoint and make up all of $\mathcal S$ in the sense that $\mathcal S = \bigcup_{q=1}^{m-1} \mathcal S_q$, up to (hyperplane) sets of $(m-1)$-dimensional Lebesgue measure zero. In fact, $\mathbf t \in\mathcal S_q $ implies that $ M({\mathbf t})=\sum_{r=1}^{q} t_{r} >0$. And the conditions for $2\le s\le q$ and for $1\le p\le m-1-q$ ensure that the sets $\mathcal S_q$ are indeed disjoint. Following again Roccaforte \cite{R} we introduce variables $\boldsymbol \tau :=(\tau_1,\ldots,\tau_{m-1})$ adapted to the just introduced sets. We define
\be\label{tau1} 
\tau_s := \sum_{r=s}^{q} t_r\, 
\ee
for $1\le s\le q$ and

\be \label{tau2}
\tau_{q+p} := \sum_{r=q+1}^{q+p} t_r\,
\ee
for $1\le p\le m-1-q$.

On the set $\mathcal S_q$  we have $\boldsymbol \tau' = \mathsf{A}^{(q)}\mathbf t' $. Here $\boldsymbol \tau'$ denotes the column tuple as the transpose of the (row) tuple $\boldsymbol\tau$ and similarly for $\mathbf t'$. And the $(m-1)\times(m-1)$ matrix $\mathsf{A}^{(q)}$ is defined in terms of its entries
\be \label{matrix A}
\mathsf{A}^{(q)}(i,j) := \left\{\begin{array}{ll}1&\mbox{ if } 1\le i\le j\le q \\1&\mbox{ if } q+1\le j\le i\\0&\mbox{ otherwise }\end{array}\right.\,.
\ee
Then we have $\det \mathsf{A}^{(q)} = 1$ and on $\mathcal S_q$ the comforting identity
\[ M(\mathbf t) = \tau_1 1_+(\tau_1) \cdots 1_+(\tau_q) 1_-(\tau_{q+1})\cdots 1_-(\tau_{m-1})\,, \quad  \mathbf t \in\mathcal S_q\,,
\] 
using the abbreviations $1_\pm$ for the indicator functions on the real line $\R$ for its two half-lines $\R_{\pm}$.

Now we consider the joint integration with respect to the $m$ variables $\xi$ and $\boldsymbol \tau$ and apply the following changes of variables. Firstly, we change $\tau_{q+1},\ldots,\tau_{m-1}$ to $-\tau_{q+1},\ldots,-\tau_{m-1}$. Clearly, the $\boldsymbol \tau$ integral is now over $\R_+^{m-1}$. Secondly, we replace $\xi$ by $\xi-(\tau_1+\tau_{m-1})/2$, and thirdly we replace $\xi$ by $-\xi$. The negative of the argument in the product of exponentials in \eqref{term} then changes according to
\be \label{exponent} m\xi^2 + \xi \sum_{i=1}^m T_i + \mfr14 \sum_{i=1}^m T_i^2 + \mfr14 \sum_{i=1}^m t_i^2 \rightsquigarrow \xi^2 + (\xi+\tau_1)^2 + \cdots+(\xi+\tau_{m-1})^2\,. 
\ee
Here and in the following we are using the notation $\rightsquigarrow$ to present the results of changes of variables efficiently, without the explicit introduction of the underlying mappings. We prove \eqref{exponent} in Appendix \ref{Misc}. The main advantage of the quadratic form \eqref{exponent} over that in the $\mathbf t$-variables is that there are no mixed terms between the $\tau$'s and the exponential can be factorized. The resulting term does not depend on $q$. This turns out to remain true with the Laguerre polynomials included as we will see next.

We perform the same changes of variables in the arguments of the Laguerre polynomials. For instance, if $q=1$, then
\begin{align*} \big(\omega+\mathrm{i}(2\xi +T_1)\big)^2 +t_1^2 &= \big(\omega+\mathrm{i}(2\xi -\tau_{m-1})\big)^2 + \tau_1^2
\\
&\rightsquigarrow \big(\omega-\mathrm{i}(2\xi + \tau_{1})\big)^2 + \tau_1^2
\\
&=\omega^2 -2\mathrm{i}\omega(2\xi+\tau_1) - (2\xi)^2 - 4\xi\tau_1\,.
\end{align*} 
Next we change $\tau_1$ to $\tau_1-\xi$ so that the last expression equals 
\[\omega^2 - 2\mathrm{i}\omega(\xi+\tau_1) - 4\xi\tau_1 = (\omega-2\mathrm{i}\xi)(\omega-2\mathrm{i}\tau_1)\,.
\]
Similarly,
\begin{align*} \big(\omega +\mathrm{i}(2\xi +T_2)\big)^2 +t_2^2 &\rightsquigarrow (\omega-2\mathrm{i}\xi)(\omega-2\mathrm{i}\tau_2)\,,
\\
\big(\omega+\mathrm{i}(2\xi +T_3)\big)^2 +t_3^2 &\rightsquigarrow (\omega-2\mathrm{i}\tau_2)(\omega-2\mathrm{i}\tau_3)\,,
\\
\big(\omega+\mathrm{i}(2\xi +T_4)\big)^2 +t_4^2 &\rightsquigarrow (\omega-2\mathrm{i}\tau_3)(\omega-2\mathrm{i}\tau_4)\,,
\\
\vdots &\phantom{\rightsquigarrow}\vdots
\\
\big(\omega+\mathrm{i}(2\xi +T_{m-1})\big)^2 +t_{m-1}^2 &\rightsquigarrow (\omega-2\mathrm{i}\tau_{m-2})(\omega-2\mathrm{i}\tau_{m-1})\,,
\\
\big(\omega+\mathrm{i}(2\xi +T_m)\big)^2 +t_m^2 &\rightsquigarrow (\omega-2\mathrm{i}\tau_{m-1})(\omega-2\mathrm{i}\tau_{1})\,.
\end{align*}
For general $q\ge 2$, see Appendix \ref{appendix B.2}. In the end, the product of the Laguerre polynomials equals
\begin{align} \label{product Laguerre} &\prod_{1\le j\le q-1}  \mathcal{L}_\ell\big((\omega_j-2\mathrm{i}\tau_j)(\omega_j-2\mathrm{i}\tau_{j+1})/2\big) \prod_{q+2\le j\le m} \mathcal{L}_\ell\big((\omega_j-2\mathrm{i}\tau_{j-1})(\omega_j-2\mathrm{i}\tau_{j})/2\big)
\\
&\times\,\mathcal{L}_\ell\big((\omega_q-2\mathrm{i}\xi)(\omega_q-2\mathrm{i}\tau_{q})/2\big) \mathcal{L}_\ell\big((\omega_{q+1}-2\mathrm{i}\xi)(\omega_{q+1}-2\mathrm{i}\tau_{q+1})/2\big) \mathcal{L}_\ell\big((\omega_q-2\mathrm{i}\tau_1)(\omega_q-2\mathrm{i}\tau_{m-1})/2\big)\,.\notag
\end{align}
The following remarkable identity will be proved in Appendix \ref{remarkable},
\begin{align} \frac{1}{\sqrt{2\pi}} &\int_\R \text d\omega \, \mathcal{L}_\ell\big((\omega-2\mathrm{i}\xi)(\omega-2\mathrm{i}\tau)/2\big) \exp(-\omega^2/4) = \sqrt{2} (2^\ell \ell!)^{-1} H_\ell(\xi) H_\ell(\tau) \label{Hermite identity}\,.
\end{align}
After performing the $m$-fold integration with respect to $\boldsymbol{\omega}$ we obtain
\begin{align*}
2^{m/2}(2^\ell \ell!)^{-1} H_\ell^2(\xi) (2^\ell \ell!)^{-1} H_\ell^2(\tau_1) \cdots (2^\ell \ell!)^{-1} H_\ell^2(\tau_{m-1})\,.
\end{align*}

To summarize, the boundary-curve term of the order $L$ equals $-L\sqrt{B} |\partial \Lambda|/(2\pi)$ times

\begin{align*} (m-1)(\sqrt{\pi} 2^{\ell} \ell!)^{-1}\int_\R &\text d\xi\,H_\ell^2(\xi) \exp(-\xi^2) \int_\xi^\infty \text d\tau_1 (\tau_1-\xi) H_\ell^2(\tau_1) \exp(-\tau_1^2) \l_\ell(\xi)^{m-2}
\\
&=-\int_\R \text d\xi\, (m-1) \l_\ell'(\xi) \l(\xi)^{m-2} \,\int_\xi^\infty \text d\tau (\tau-\xi) H_\ell^2(\tau) \exp(-\tau^2) 
\\
&=-\int_\R \text d\xi\, \frac{\text d}{\text d\xi} (\l_\ell(\xi)^{m-1} - 1) \,\int_\xi^\infty \text d\tau (\tau-\xi) H_\ell^2(\tau) \exp(-\tau^2) 
\\
&=-\left[(\l_\ell(\xi)^{m-1} - 1)\int_\xi^\infty \text d\tau (\tau-\xi) H_\ell^2(\tau) \exp(-\tau^2) \right]_{-\infty}^\infty 
\\
&+\int_\R \text d\xi\, (\l_\ell(\xi)^{m-1} - 1) \frac{\text d}{\text d\xi} \int_\xi^\infty \text d\tau (\tau-\xi) H_\ell^2(\tau) \exp(-\tau^2)
\\
&=-\int_\R \text d\xi\, (\l_\ell(\xi)^{m-1} - 1) \l_\ell(\xi)\,.
\end{align*}

\medskip

Finally, we turn to the leading area term of the order $L^2$ in \eqref{moment2},
$$ L^2 B |\L|\,\int_{\R^{m-1}}\text d\mathbf t\, g(t_1)\cdots g(t_{m-1}) g(t_m) \, I_m(\mathbf t) \,.
$$
Here, we use \eqref{U-turn} for $I_m(\mathbf t)$ and switch to the variables $\tau_1,\ldots,\tau_{m-1}$ from \eqref{tau1} and \eqref{tau2} for $q=1$, in all of $\R^{m-1}$. We perform the same shifts in $\xi$ and in the $\tau$'s. Then the area term turns into 
\begin{align*} 
L^2 B \frac{|\Lambda|}{2\pi} \, 
\left(\int_\R \text d\xi \,(\sqrt{\pi} 2^\ell \ell!)^{-1} 
H_\ell^2(\xi) \exp(-\xi^2)\right)^m = L^2 B \frac{|\Lambda|}{2\pi}
\end{align*}
by the normalization of the Hermite functions. 
Alternatively, the leading term can be obtained 
by replacing the $x$-integral in \eqref{moment} by $|\Lambda|$. 
The remaining $\mathbf{y}$-integration yields $B/(2\pi)$. This finishes the proof of Lemma \ref{lem:moments}.
\end{proof}

The next lemma provides the basis for the proof of Theorem \ref{thm:ground state smooth f}.

\begin{lemma} \label{lem:general moments} Under the same assumptions as in Lemma \ref{lem:moments} we have for any $n\in\N_0$
\be \tr {\rm P}_{\le n}(L\L)^m = L^2 B \,\frac{|\L|}{2\pi}\,(n+1) + L\sqrt{B}|\p\L| \, \int_\R \frac{{\mathrm d} \xi}{2\pi}\,[\tr\mathcal K_{n,\xi}^m - \,\tr \mathcal K_{n,\xi}] + \mathcal{O}(1)\,,
\ee
as $L\to\infty$. 
\end{lemma}

\begin{proof} By the same arguments as in the beginning of 
the proof of Lemma \ref{lem:moments}, 
the projection $\text P_{\le n}(\L)^m$ is 
a trace-class operator and its trace can be calculated by 
integrating the diagonal of the $m$-fold iterated integral 
kernel of $\text P_{\le n}(\L)$. Again, the case $m=1$ 
is then obvious and we only need to consider the case $m\ge 2$. We recall that $\text P_{\le n}(x,y) = \sum_{0\le \ell\le n}p_\ell(x,y)$. So in the proof of Lemma \ref{lem:moments} we simply have to replace $\mathcal{L}_\ell$ with $\mathcal{L}_{\le n} = \sum_{0\le \ell\le n}\mathcal{L}_\ell = \mathcal{L}_n^{(1)}$. For instance, expression \eqref{product Laguerre} is replaced with the expression
\begin{align*} \prod_{1\le j\le q-1} & \mathcal{L}_{\le n}\big((\omega_j-2\mathrm{i}\tau_j)(\omega_j-2\mathrm{i}\tau_{j+1})/2\big) \prod_{q+2\le j\le m} \mathcal{L}_{\le n}\big((\omega_j-2\mathrm{i}\tau_{j-1})(\omega_j-2\mathrm{i}\tau_{j})/2\big)
\\
&\times \mathcal{L}_{\le n}\big((\omega_q-2\mathrm{i}\xi)(\omega_q-2\mathrm{i}\tau_{q})/2\big) \mathcal{L}_{\le n}\big((\omega_{q+1}-2\mathrm{i}\xi)(\omega_{q+1}-2\mathrm{i}\tau_{q+1})/2\big) 
\\
&\times \mathcal{L}_{\le n}\big((\omega_q-2\mathrm{i}\tau_1)(\omega_q-2\mathrm{i}\tau_{m-1})/2\big)\,.
\end{align*}
We multiply this expression by $(2\pi)^{-m/2}\prod_{j=1}^{m} \exp(-\omega_j{^2}/4)$ and integrate with respect to $\boldsymbol{\omega}$ over $\R^m$ by using \eqref{Hermite identity}. This yields 
\begin{align*} 2^{m/2}&\prod_{1\le j\le q-1}\sum_{\ell_j=0}^n\big(2^{\ell_j} \ell_j!\big)^{-1} H_{\ell_j}(\tau_j) H_{\ell_j}(\tau_{j+1})\prod_{q+2\le j\le m-1}\sum_{\ell_j=0}^n\big(2^{\ell_j} \ell_j!\big)^{-1} H_{\ell_j}(\tau_{j-1}) H_{\ell_j}(\tau_{j})
\\
&\times\sum_{\ell_q=0}^n\big(2^{\ell_q} \ell_q!\big)^{-1} H_{\ell_q}(\xi) H_{\ell_q}(\tau_{q})\sum_{\ell_{q+1}=0}^n\big(2^{\ell_{q+1}} \ell_{q+1}!\big)^{-1} H_{\ell_{q+1}}(\xi) H_{\ell_{q+1}}(\tau_{q+1})
\\
&\times\sum_{\ell_m=0}^n\big(2^{\ell_m} \ell_m!\big)^{-1} H_{\ell_m}(\tau_{1}) H_{\ell_m}(\tau_{m-1})\,.
\end{align*}
To pause for a moment, the term of the order $L$ equals $-\sqrt{B}L|\partial\Lambda|/(2\pi)$ times
\begin{align*}\sum_{q=1}^{m-1}&\int_\R\frac{\text d\xi}{2\pi} \,\exp(-\xi^2)\int_\xi^\infty \text d\tau_1\,(\tau_1-\xi) \,\exp(-\tau_1^2)\int_\xi^\infty \text d\tau_2\,\exp(-\tau_2^2)\cdots \int_\xi^\infty \text d\tau_{m-1}\,\exp(-\tau_{m-1}^2)
\\
&\times \pi^{-m/2}\prod_{1\le j\le q-1}\sum_{\ell_j=0}^n\big(2^{\ell_j} \ell_j!\big)^{-1} H_{\ell_j}(\tau_j) H_{\ell_j}(\tau_{j+1})\prod_{q+2\le j\le m-1}\sum_{\ell_j=0}^n\big(2^{\ell_j} \ell_j!\big)^{-1} H_{\ell_j}(\tau_{j-1}) H_{\ell_j}(\tau_{j})
\\
&\times\sum_{\ell_q=0}^n\big(2^{\ell_q} \ell_q!\big)^{-1} H_{\ell_q}(\xi) H_{\ell_q}(\tau_{q})\sum_{\ell_{q+1}=0}^n\big(2^{\ell_{q+1}} \ell_{q+1}!\big)^{-1} H_{\ell_{q+1}}(\xi) H_{\ell_{q+1}}(\tau_{q+1})
\\
&\times\sum_{\ell_m=0}^n\big(2^{\ell_m} \ell_m!\big)^{-1} H_{\ell_m}(\tau_{1}) H_{\ell_m}(\tau_{m-1})\,.
\end{align*}
We include the factors of $\pi$ into the terms $(2^\ell \ell!)^{-1/2}$, split them in halves, and combine them with each corresponding factor $H_\ell$. In accordance with that we define 
$$ \l_{\ell_i,\ell_j}(\xi) := \int_\xi^\infty \text d\tau \, \big(\sqrt{\pi} 2^{\ell_i} \ell_i!\big)^{-1/2} H_{\ell_i}(\tau) \big(\sqrt{\pi} 2^{\ell_j} \ell_j!\big)^{-1/2} H_{\ell_j}(\tau)\, \exp(-\tau^2)\,.
$$
Then the summand for $q=1$ can be written in the form
\begin{align} -\sum_{0\le \ell_1,\ldots,\ell_m\le n} &\int_\R \frac{\text d\xi}{2\pi}\, \Big[\frac{\text d}{\text d \xi}\l_{\ell_1,\ell_2}(\xi)\Big] \, \l_{\ell_2,\ell_3}(\xi) \cdots \l_{\ell_{m-1},\ell_m}(\xi)
\\
&\times \int_\xi^\infty \text d\tau_1 \,(\tau_1-\xi) \big(\sqrt{\pi} 2^{\ell_1} \ell_1!\big)^{-1/2} H_{\ell_1}(\tau_1) \big(\sqrt{\pi}2^{\ell_m} \ell_m!\big)^{-1/2} H_{\ell_m}(\tau_1)\, \exp(-\tau_1^2)\,.\notag
\end{align}
For $q=2$ we get the term
\begin{align} -\sum_{0\le \ell_1,\ldots,\ell_m\le n}&\int_\R \frac{\text d\xi}{2\pi} \l_{\ell_1,\ell_2}(\xi)\Big[\frac{\text d}{\text d \xi} \l_{\ell_2,\ell_3}(\xi)\Big] \, \l_{\ell_3,\ell_4}(\xi)\cdots \l_{\ell_{m-1},\ell_m}(\xi)
\\
&\times \int_\xi^\infty \text d\tau_1 \,(\tau_1-\xi) \big(\sqrt{\pi} 2^{\ell_1} \ell_1!\big)^{-1/2} H_{\ell_1}(\tau_1) \big(\sqrt{\pi}2^{\ell_m} \ell_m!\big)^{-1/2} H_{\ell_m}(\tau_1)\, \exp(-\tau_1^2)\,,\notag
\end{align}
and similarly for $q=3,\dots,m-1$. By summing over all $q$ we obtain 
\begin{align}-\sum_{0\le \ell_1,\ldots,\ell_m\le n} &\int_\R \frac{\text d\xi}{2\pi}\, \frac{\text d}{\text d \xi} \Big[\l_{\ell_1,\ell_2}(\xi) \l_{\ell_2,\ell_3}(\xi) \cdots \l_{\ell_{m-1},\ell_m}(\xi)\Big]
\\
&\times \int_\xi^\infty \text d\tau_1 \,(\tau_1-\xi) \big(\sqrt{\pi} 2^{\ell_1} \ell_1!\big)^{-1/2} H_{\ell_1}(\tau_1) \big(\sqrt{\pi}2^{\ell_m} \ell_m!\big)^{-1/2} H_{\ell_m}(\tau_1)\, \exp(-\tau_1^2)\,.\notag
\end{align}
Inside the derivative with respect to $\xi$ we subtract the constant $C_{\ell_1,\ldots,\ell_m} := \l_{\ell_1,\ell_2}(-\infty) \l_{\ell_2,\ell_3}(-\infty)$ $ \cdots \l_{\ell_{m-1},\ell_m}(-\infty)$ so that we can integrate by parts. Then we get
\begin{align*} -\int_\R \frac{\text d\xi}{2\pi}\,& \frac{\text d}{\text d \xi} 
\Big[\l_{\ell_1,\ell_2}(\xi) \l_{\ell_2,\ell_3}(\xi) \cdots \l_{\ell_{m-1},\ell_m}(\xi)-C_{\ell_1,\ldots,\ell_m}\Big]
\\
&\times \int_\xi^\infty \text d\tau_1 \,(\tau_1-\xi) \big(\sqrt{\pi} 2^{\ell_1} \ell_1!\big)^{-1/2} H_{\ell_1}(\tau_1) \big(\sqrt{\pi}2^{\ell_m} \ell_m!\big)^{-1/2} H_{\ell_m}(\tau_1)\, \exp(-\tau_1^2)
\\
&=-\int_\R \frac{\text d\xi}{2\pi}\,\Big[\l_{\ell_1,\ell_2}(\xi) \l_{\ell_2,\ell_3}(\xi) \cdots \l_{\ell_{m-1},\ell_m}(\xi) -C_{\ell_1,\ldots,\ell_m}\Big] \l_{\ell_1,\ell_m}(\xi)
\\
&=-\int_\R \frac{\text d\xi}{2\pi}\,\Big[\l_{\ell_1,\ell_2}(\xi) \l_{\ell_2,\ell_3}(\xi) \cdots \l_{\ell_{m-1},\ell_m}(\xi)\l_{\ell_m,\ell_1}(\xi) - C_{\ell_1,\ldots,\ell_m} \l_{\ell_1,\ell_m}(\xi)\Big]\,.
\end{align*}
By \eqref{orthog} we know that $C_{\ell_1,\ldots,\ell_m} = \d_{\ell_1,\ell_2}\d_{\ell_2,\ell_3}\cdots\d_{\ell_{m-1},\ell_m}$. Now we sum over $\ell_1,\ldots,\ell_m$ and use the Christoffel--Darboux formula (of the years 1858 and 1878)
\[\sum_{\ell=0}^n (2^\ell \ell!)^{-1} H_\ell(\tau) H_\ell(\tau') = (2^{n+1}n!)^{-1}\frac{H_n(\tau') H_{n+1}(\tau) - H_n(\tau) H_{n+1}(\tau')}{\tau-\tau'}\quad\mbox{ if } \tau\not=\tau' 
\]
and 
\be \label{CD}
\sum_{\ell=0}^n (2^\ell \ell!)^{-1} H_\ell(\tau)^2 = (2^{n+1} n!)^{-1} \big[H_{n+1}^2(\tau) - H_{n}(\tau) H_{n+2}(\tau)\big]\,. 
\ee
Then, for instance,
\begin{align*} \sum_{0\le \ell_2\le n} & \l_{\ell_1,\ell_2}(\xi)\l_{\ell_2,\ell_3}(\xi)\
\\
&=\big(\sqrt{\pi}2^{\ell_1} \ell_1!\big)^{-1/2} \big(\sqrt{\pi}2^{\ell_3} \ell_3!\big)^{-1/2} \int_{[\xi,\infty)^2} \text d\tau_1 \text d\tau_2\, H_{\ell_1}(\tau_1)H_{\ell_3}(\tau_2) 
\\
&\times \sum_{0\le \ell_2\le n} \big(\sqrt{\pi}2^{\ell_2} \ell_2!\big)^{-1}  H_{\ell_2}(\tau_1) H_{\ell_2}(\tau_2) \, \exp(-\tau_1^2-\tau_2^2)
\\
&=\big(\sqrt{\pi}2^{\ell_1} \ell_1!\big)^{-1/2} \big(\sqrt{\pi}2^{\ell_3} \ell_3!\big)^{-1/2} (\sqrt{\pi}2^{n+1}n!)^{-1}
\\
&\times \int_{[\xi,\infty)^2} \text d\tau_1 \text d\tau_2\,H_{\ell_1}(\tau_1)H_{\ell_3}(\tau_2)\,\frac{H_n(\tau_2) H_{n+1}(\tau_1) - H_n(\tau_1) H_{n+1}(\tau_2)}{\tau_1-\tau_2} \,\exp(-\tau_1^2-\tau_2^2)\,.
\end{align*}
Performing also the summations over $\ell_1,\ell_3,\dots,\ell_m$ yields
\begin{align} \big(\sqrt{\pi}& 2^{n+1} n!)^{-m} \int_{[\xi,\infty)^{m}} \text d\boldsymbol \tau\, \exp(-\boldsymbol \tau^2)\,\frac{H_n(\tau_2) H_{n+1}(\tau_1) - H_n(\tau_1) H_{n+1}(\tau_2)}{\tau_1-\tau_2}\notag
\\
&\times \,\frac{H_n(\tau_3) H_{n+1}(\tau_2) - H_n(\tau_2) H_{n+1}(\tau_3)}{\tau_2-\tau_3}\,\cdots\,\frac{H_n(\tau_{1}) H_{n+1}(\tau_m) - H_n(\tau_m) H_{n+1}(\tau_{1})}{\tau_{m}-\tau_1}\notag
\\
&=:\l_{\le n,m}(\xi)\,.\label{lambda m}
\end{align}
By \eqref{CD} we also find
\begin{align}
\sum_{0\le \ell_1,\ldots,\ell_m\le n} C_{\ell_1,\ldots,\ell_m} \l_{\ell_1,\ell_m}(\xi) &= \sum_{0\le \ell\le n}  (\sqrt{\pi}2^{\ell}\ell!)^{-1} \int_\xi^\infty \text d\tau \, H_\ell(\tau)^2 \,\exp(-\tau^2)\notag
\\
&=\big(\sqrt{\pi} 2^{n+1} n!)^{-1} \int_\xi^\infty \text d\tau\, \big[H_{n+1}^2(\tau) - H_{n}(\tau) H_{n+2}(\tau)\big]\,\exp(-\tau^2)\notag
\\
&=:\l_{\le n,1}(\xi)\,.\label{lambda 1}
\end{align}
Now we apply \eqref{CD} directly to the integral kernel of the operator $\mathcal K_{n,\xi}$ as defined in \eqref{opK}. Then we get

\begin{align} \label{op: K}
\mathcal K_{n,\xi}(\tau,\tau') &= \frac{\exp(-(\tau^2+\tau'^2)/2)}{\sqrt{\pi} 2^{n+1}n!} \begin{cases}
\dfrac{H_n(\tau') H_{n+1}(\tau) - H_n(\tau) H_{n+1}(\tau')}{\tau-\tau'}&\mbox{ if } \tau\not=\tau'\\ 
H_{n+1}^2(\tau) - H_{n}(\tau) H_{n+2}(\tau)&\mbox{ if } \tau=\tau'\end{cases}
\end{align}
whenever $\tau\ge \xi$ and $\tau'\ge\xi$ and zero otherwise. 
By comparing this with the just given definitions \eqref{lambda m} and \eqref{lambda 1} of $\l_{\le n,m}(\xi)$ we arrive at the relation $\l_{\le n,m}(\xi) = \tr\mathcal K_{n,\xi}^m$ for all $m\in\N$ as claimed in Lemma \ref{lem:general moments} for the sub-leading term of the order $L$. Finally, we turn to the leading term of the order $L^2$. It equals 
\begin{align*} L^2B|\L| &\int_\R \frac{\text d\xi}{2\pi} \, \exp(-\xi^2) \int_\R \text d\tau_1 \exp(-\tau_1^2)\cdots \int_\R \text d\tau_{m-1} \exp(-\tau_{m-1}^2)
\\
&\times\,\sum_{\ell_1=0}^n (\sqrt{\pi}2^{\ell_1}\ell_1!)^{-1} H_{\ell_1}(\xi)H_{\ell_1}(\tau_1) \,\sum_{\ell_2=0}^n (\sqrt{\pi} 2^{\ell_2}\ell_2!)^{-1} H_{\ell_2}(\xi)H_{\ell_2}(\tau_2)\cdots
\\
&\times\,\sum_{\ell_m=0}^n (\sqrt{\pi}2^{\ell_m}\ell_m!)^{-1} H_{\ell_m}(\tau_1)H_{\ell_m}(\tau_{m-1})\,.
\end{align*}
By the orthogonality \eqref{orthog} this yields $L^2B(n+1)|\L|/(2\pi)$ as claimed.
\end{proof}

\section{From smooth functions to the entropy functions}

In this section we build on Theorem \ref{thm:LL smooth f} and Theorem \ref{thm:ground state smooth f} with a suitable function $f$ to derive the precise leading asymptotic growth of the local ground-state entropy with arbitrary R\'enyi index $\a > 0$. While the case $\a >1$ is rather straightforward, non-smoothness in the case $\a\leq 1$ requires considerable attention. In the first subsection we define the local ground-state entropies and present our main result and related results. The second subsection prepares the ground for getting from smooth functions to the non-smooth functions needed in the case $\a\leq 1$. Proofs of our results are then given in the third subsection.

\subsection{Definitions and results}
For a real parameter $\a>0$ we define the \textit{$\a$-R\'enyi entropy function} $h_\a:[0,1]\to[0,\ln(2)]$ by
\begin{align} \label{def:entropy}
h_\a(t) := &\ \frac{1}{1-\a}\ln\big(t^\a + (1-t)^\a\big)\,,\quad\ \a\not = 1,\notag\\[0.2cm]
h_1(0):= & \ h_1(1):=0\,,\, h_1(t) :=-t\ln(t) - (1-t)\ln(1-t)\quad\mbox{ if } t\notin \{0,1\} \,
\end{align}
and recall \eqref{landau} as well as \eqref{sr} for the Landau Hamiltonian $\text H$.

Then the positive number
\be\label{entropy}
S_\a(\L):=\tr h_\a(\mathds{1}_\L \Theta(\mu \mathds{1}- \text H )\mathds{1}_\L ) = \tr h_\a(\text P_{\le \nu}(\L))=\tr \mathds{1}_\L h_\a(\text P_{\le \nu}(\L)) \mathds{1}_\L 
\ee
is the \textit{local $\a$-R\'enyi ground-state entropy} (with chemical potential $\mu\ge B$), see \eqref{ss} and \cite{HLS}. Here, the integer $\nu\in\N_0$ is the integer part of $(\mu/B-1)/2$ as defined already below \eqref{gs}. In particular, $S(\L) := S_1(\L)$ is the local von Neumann ground-state entropy mentioned in the abstract. The third equality in \eqref{entropy} is due to $h_\a(0)=0$. 

Finiteness of the local $\a$-R\'enyi ground-state entropy is guaranteed by 

\begin{lemma}\label{lem:finite}
Let $\L\subset\R^2$ be a bounded Borel set and $\mu\geq B$. Then $S_\alpha(\Lambda) <\infty$ for any $\a>0$.
\end{lemma}

The proof is given in Subsection \ref{proof} after certain preparations in Subsection \ref{sv}. The next theorem gives the precise asymptotic growth. It is our main result.

\begin{thm}[\bf{Asymptotics of the local R\'enyi ground-state entropies}]\label{entropy of ground state}
Let $\L\subset\R^2$ be a bounded $\mathsf C^3$-region and let the chemical potential satisfy $\mu\geq B$. Then the local $\a$-R\'enyi ground-state entropy \eqref{entropy} obeys
\be\label{main} 
S_\a(L\L) = L\sqrt{B}\,|\p\L|\,{\sf M}_{\le \nu}(h_\a)+ o(L)\,,
\ee
as $L\to\infty$. The asymptotic coefficient ${\sf M}_{\le\nu}(h_\alpha)$ is given by \eqref{mleell:eq} with $n=\nu$. It is finite and positive. 
\end{thm}

The proof is given in Subsection \ref{proof}. It builds on Lemma \ref{lem:finite}, Theorem~\ref{thm:ground state smooth f}, and Subsection \ref{sv}.

\begin{remarks} 
\begin{enumerate}
\item[(i)] The coefficient ${\sf M}_{\le\nu}(h_\alpha)$ in \eqref{main} is in general not easy to calculate. The simplest case occurs when $\nu=0$. Then we have
\be\label{coeff}
 {\sf M}_{\le 0}(h_\alpha) = {\sf M}_{0}(h_\alpha) = \int_\R \frac{\text d\xi}{2\pi}\, h_\a(\l_0(\xi))
\ee
with $\l_0(\xi) = {\pi}^{-1/2}\int_\xi^\infty \text d t \,\exp(-t^2)$ being $1/2$ of the complementary error function. The coefficient \eqref{coeff} was found in \cite{RS1} for $\a=1$ and special regions $\L$. The first proof of \eqref{main} for $\a=1$, $\nu=0$ (equivalently, $\ell=0$ in \eqref{main_l}), $L^2 B\in\N$, and general bounded $\mathsf C^\infty$-regions is due to Charles and Estienne in \cite{CE}. A numerical computation gives ${\sf M}_{0}(h_1) = 0.203\dots$, also in agreement with \cite{RS1}.

\item[(ii)] 
In the zero-field case $B=0$, the 
leading term of the local R\'enyi ground-state 
entropy depends on its index $\a$ simply through the pre-factor $(1+\a)/\a$, see \cite{LSS1}. 
A numerical computation shows that in the case $B\not = 0$ the dependence on $\a$ is not so simple. 

\end{enumerate}
\end{remarks}

We are going to define a local $\a$-R\'enyi ground-state 
entropy also for the simpler situation where the 
Landau Hamiltonian $\text H$ is \textit{restricted} 
(or ``projected'') from the outset to a single Landau-level 
eigenspace $\text P_\ell\text L^2(\R^2)$ with arbitrary index $\ell\in\N_{0}$. 
This restriction means that $\text H$ is replaced with 
$\text P_{\ell} \text H \text P_{\ell}$ and 
similarly for related operators, confer, for example, \cite{HLW}. 
Then the corresponding local(alized) Fermi projection is in analogy to \eqref{ss} given by
\be\label{fermi_l}
\mathds{1}_\L\text P_{\ell} \Theta(\mu\text P_{\ell}\mathds{1} \text P_{\ell}  -\text P_{\ell} \text H \text P_{\ell}) \text P_{\ell}\mathds{1}_\L= \Theta(\mu-(2\ell+1)B)\text P_{\ell}(\L) = \text P_{\ell}(\L) \,,\quad \mu\ge(2\ell+1)B\,.
\ee
We ignore the case $\mu <(2\ell+1)B$, because then the local Fermi projection \eqref{fermi_l} is the zero operator. In analogy to \eqref{entropy} we now define for each $\ell\in\N_0$ the positive number
\be\label{entropy_l}
S_{\alpha,\ell}(\L):=\tr h_\alpha (\mathds{1}_\L \text P_{\ell} \Theta(\mu \text P_{\ell} -\text P_{\ell} \text H \text P_{\ell}) \text P_{\ell}\mathds{1}_\L )=\tr h_\alpha(\text P_{\ell}(\L))=\tr \mathds{1}_\L h_\alpha(\text P_{\ell}(\L))\mathds{1}_\L
\ee
and call it the \textit{local $\alpha$-R\'enyi ground-state entropy of 
the $\ell$th Landau level} (with chemical potential $\mu\ge(2\ell+1)B$). 
Obviously, we have $S_{\alpha,0}(\L)=S_\a(\L)$ if $0<B\le \mu<3B$. 
Along with Theorem~\ref{entropy of ground state} the following result holds 

\begin{thm}[\bf{Asymptotics of the local R\'enyi  ground-state entropies of the $\ell$th Landau level}]\label{entropy of ground state Landau}
Let $\L\subset\R^2$ be a bounded $\mathsf C^3$-region, $\ell\in\N_0$, and $\mu\geq (2\ell+1)B$. Then the local $\alpha$-R\'enyi ground-state entropy of the $\ell$th Landau level \eqref{entropy_l} obeys
\be\label{main_l} 
S_{\a,\ell}(L\L) = L\sqrt{B}\,|\p\L|\,{\sf M}_{\ell}(h_\a)+ o(L)\,,
\ee
as $L\to\infty$. The asymptotic coefficient ${\sf M}_{\ell}(h_\alpha)$ is given by \eqref{mell:eq}. It is finite and positive. 
\end{thm}
The remark immediately below Theorem \ref{entropy of ground state} applies analogously to this theorem when \eqref{RS:general} is replaced with \eqref{RS}. The proof of Theorem~\ref{entropy of ground state Landau} builds on Theorem~\ref{thm:LL smooth f} and repeats the proof of Theorem \ref{entropy of ground state} in Subsection \ref{proof}
with $\text P_{\le \nu}$ replaced by $\text P_\ell$.

The next subsection contains estimates being crucial for the proof of the above results.

\subsection{Estimates for singular values}\label{sv}

We begin by introducing certain operators ${\rm T}_{r, R}$ on $\mathrm{L}^2(\mathbb R^2)$. To this end, we denote by ${\sf D}(x, R)\subset \mathbb R^2$ the open disk of radius $R>0$, centered at the point $x\in \mathbb R^2$ and abbreviate $\mathds{1}_R := \mathds{1}_{{\sf D}(0, R)}$ and similarly with $R$ replaced by $r>0$. Then we define the operators
\begin{align}\label{trr:eq}
{\rm T}_{r, R} := {\rm T}^{(\ell)}_{r, R}:=\mathds{1}_r \text P_\ell \big(\mathds{1} - \mathds{1}_R \big)\,,\quad {\rm T}_{r, 0}  :=  \mathds{1}_r \text P_\ell\,,\quad \ell\in \N_0\, .
\end{align}
Here we assume that the magnetic-field strength has been ``scaled out'', so that $B=1$ in formula \eqref{Landau_kernel:eq}. We interpret ${\rm T}_{r,R}$ as an operator from $\mathrm{L}^2(\mathbb R^2)$ into $\mathrm{L}^2({\sf D}(0,r))$.

In order to estimate the singular values of ${\rm T}_{r,R}$ we recall the short compilation below Lemma \ref{mleell:lem} and a classical result due to Birman and Solomyak, see \cite[Theorem 4.7]{BS}. We quote the required fact in a form adapted to our purposes.

\begin{prop}\label{BS:prop}
Let ${\rm Z}:\mathrm{L}^2(\mathbb R^2)\to \mathrm{L}^2({\sf D}(0,r))$ be an integral operator defined by a complex-valued kernel $\mathcal Z(x, y)$ obeying
\begin{align*}
N_\g(\mathcal Z) := \bigg[\sum_{0\le s, t\le \g}
\int_{\mathbb R^2} \mathrm{d} y\int_{{\sf D}(0, r)}\mathrm{d}x \,\bigg|\frac{\partial^s}{\partial x_1^s} \frac{\partial^{t}}{\partial x_2^{t}} \mathcal Z(x,y)\bigg|^2\bigg]^{\frac{1}{2}}<\infty\,,
\end{align*}
for some $\g\in \mathbb N_0$. Then the singular values $s_n({\rm Z})$ of ${\rm Z}$ satisfy the bound
\begin{align*}
s_n({\rm Z})\le C n^{-\frac{1+\g}{2}} N_\g(\mathcal Z)\,,\quad n\in\N\,,
\end{align*}
with a positive constant $C$ dependent on $r$ but independent of the kernel $\mathcal Z$.
\end{prop}

\begin{lemma} \label{sp:lem}
The operator ${\rm T}_{r, 0}$ belongs to the Schatten--von Neumann class $\mathfrak S_p$ for all $p>0$. Moreover, if $R >r$, then 
\begin{align}\label{sp:eq}
\| {\rm T}_{r, R}\|_p \le C \,\exp\big(-(R-r)^2/8\big)\,,
\end{align}
with some constant $C$ dependent on $r$ but independent of $R$.
\end{lemma}

\begin{proof} In order to apply Proposition \ref{BS:prop}, we estimate for $\|x\| < r$ and $s,t\in\N_0$ with $0\le s,t\le\gamma$:
\begin{align*}
\bigg|\frac{\partial^s}{\partial x_1^s} \frac{\partial^{t}}{\partial x_2^{t}} p_\ell(x, y)\bigg|\le C_{\gamma, \ell} (1+\|y\|)^\gamma (1+ \|x-y\|^2)^\ell \exp\big(-\|x-y\|^2/4\big)\,,
\end{align*} 
with a constant $C_{\gamma, \ell}$ depending on $r$. Here we have used the fact that $\mathcal{L}_\ell$ is a polynomial of degree $\ell$. For $R > r$ and $y\notin {\sf D}(0, R)$ we conclude that 
\begin{align*}
\bigg|\frac{\partial^s}{\partial x_1^s} \frac{\partial^{t}}{\partial x_2^{t}} p_\ell(x, y)\bigg|\le C_{\gamma, \ell} (1+\|y\|)^\gamma (1+ \|x-y\|^2)^\ell \,\exp\big(-(R-r)^2/8\big)\,\exp\big(-\|x-y\|^2/8\big)\,.
\end{align*} 
Thus the integral kernel 
$\mathcal T_{r, R}(x,y) := 1_r(x) p_\ell(x, y) (1 - 1_R(y))$ 
of ${\rm T}_{r, R} $ satisfies, for any  $\g\in \mathbb N$, the bounds
\begin{align*}
N_\g(\mathcal T_{r, R})\le &\  C_{\g, \ell}\,, \ {\rm for \ any}\ R\ge 0\,,\\
N_\g(\mathcal T_{r, R})\le &\  C_{\g, \ell} \,\exp\big(-(R-r)^2/8\big)\,,\ {\rm if} \ R>r\,,
\end{align*}
where the constant $C_{\g, \ell}$ is independent of $R$. For an arbitrary $p>0$ we now take $\g>2p^{-1}-1$. 
Then by Proposition \ref{BS:prop}, ${\rm T}_{r, R}\in \mathfrak S_p$ for all $R\ge 0$ 
and the bound \eqref{sp:eq} holds for $R>r$.
\end{proof}

Lemma \ref{sp:lem} is an important ingredient to bound quasi-norms of the operator $\mathds{1}_{L\L}\text P_\ell (\mathds{1}-\mathds{1}_{L\L})$. Although our main result Theorem \ref{entropy of ground state} is proved for a bounded $\mathsf C^3$-region $\Lambda$, the next theorem even holds for a bounded Lipschitz region. Here, the boundary curve $\partial\Lambda$ of $\L$ is Lipschitz continuous. 

\begin{thm}\label{Hank:thm}
Let $\Lambda\subset \mathbb R^2$ be a bounded Lipschitz region and $\ell\in\N_0$. Moreover, let $p\in {(0, 1]}$ and $L_0 >0$ finite. Then there exists a constant $C$, depending only on $\L$ and $L_0$, such that for any $L\ge L_0$,
\begin{align}\label{Hank:eq}
\|\mathds{1}_{L\Lambda} {\rm P}_\ell (\mathds{1}-\mathds{1}_{L\Lambda})\|_p^p \le CL\,.
\end{align}
\end{thm}

\begin{proof} 
We begin with two useful observations. Firstly, we note that
$\|\mathds{1}_{L\Lambda} {\rm P}_\ell (\mathds{1}-\mathds{1}_{L\Lambda})\|_p\le \|\text T_{Lr, 0}\|_p$ by \eqref{holder} with some $r >0$, where the operator $\text T_{r, 0}$ is defined in \eqref{trr:eq}. 
Thus, by Lemma \ref{sp:lem}, for every fixed $L$ the left-hand side of the claim \eqref{Hank:eq} is finite. 
Consequently, it suffices to prove \eqref{Hank:eq} for $L\ge L_0$ with an arbitrary choice of the finite $L_0$. 

Secondly, denoting $\Lambda' := L_0\L$, we can rewrite \eqref{Hank:eq} as follows:
\begin{align*}
\|\mathds{1}_{L\Lambda'} {\rm P}_\ell (\mathds{1}-\mathds{1}_{L\Lambda'})\|_p^p \le CL\,, \quad L\ge 1\,,
\end{align*}
where $C$ depends on $\L$ and the arbitrary $L_0$. 

Now we can proceed with the proof. We cover $\L$ by finitely many disks ${\sf D}(x_k, r_k)$ 
such that either
\begin{enumerate}
\item 
$x_k\in \p\L$ and inside each disk ${\sf D}(x_k, 8r_k)$, with an appropriate choice of coordinates, the domain $\L$ is given locally by the epigraph of a Lipschitz function (see below), or 
\item 
${\sf D}(x_k, r_k)\subset \L$ and ${\rm dist}\big({\sf D}(x_k, r_k), \L^\complement\big)>0$, where $\Lambda^\complement:=\R^2\setminus\Lambda$ denotes the complement of $\Lambda$.
\end{enumerate}
It is clear that we may assume that all radii $r_k$ are equal to each other. Moreover, by replacing $\L$ with $\L' = L_0\Lambda$ with $L_0 = r_k^{-1}$, we may assume that $r_k=1$. This equality holds throughout the proof.

\underline{Case (1):}
We fix one disk ${\sf D} := {\sf D}(x_k, 1)\subset \mathbb R^2$, $x_k\in\p\L$, 
and denote $\widetilde{\sf D} := {\sf D}(x_k, 8)$. 
Let 
$\Phi: \mathbb R\to \mathbb R$ be a Lipschitz function such that
\begin{align*}
\L\cap \widetilde{\sf D} = \{x = (x', x'')\in\mathbb R^2: x'' > \Phi(x')\}\cap \widetilde{\sf D}\,.
\end{align*} 
By $M\ge 0$ we denote the Lipschitz constant for $\Phi$, i.e. 
\begin{align*}
|\Phi(t) - \Phi(u)|\le M|t-u|\,,\quad t, u\in\mathbb R\,.
\end{align*}
It is clear that
\begin{align*}
(L\L\cap L{\sf D})\subset (L\L\cap L\widetilde{\sf D}) = \{x = (x', x''): x'' > \Phi_L(x')\}\cap L\widetilde{\sf D}\,,
\quad \Phi_L(t):= L\Phi(tL^{-1})\,, \quad t\in \mathbb R\,,
\end{align*}
and that the Lipschitz constant for $\Phi_L$ also equals $M$.  
Without loss of generality we may assume that ${\sf D} = {\sf D}(0, 1)$ and $\Phi(0) = 0$. 

Now we construct a  
covering of $L\L\cap L{\sf D}$ by open disks. 
Let ${\sf D}_{jk}$ be a disk of radius $1$, centered 
at the point $z_{j, k} := (j/2, k/2)\in L{\sf D}$, $(j, k)\in\mathbb Z^2$. 
Clearly, such disks form an open covering for $L{\sf D}$. 
To extract a convenient covering for $L\L\cap L{\sf D}$, we define two index sets: 
\begin{align*}
I_1 := &\ \{(j, k)\in \mathbb Z^2: k/2 \ge \Phi_L(j/2)+ 2\langle  M\rangle,\ z_{jk}\in L{\sf D}\}\,,
\quad \langle M\rangle := \sqrt{1+M^2}\,,\\ 
I_2 := &\ \{(j, k)\in \mathbb Z^2: |k/2 - \Phi_L(j/2)| < 2\langle  M\rangle,\ z_{jk}\in L{\sf D}\}\,.
\end{align*}
Since $\Phi$ is Lipschitz, the number of indices in $I_2$ obeys $|I_2| \le CL$ with a constant independent of $L$. 
The disks ${\sf D}_{jk}$, $(j, k)\in I_1\cup I_2$ form a covering of the intersection $L\Lambda\cap L{\sf D}$. 
Observe that for every point $x = (x', x'')\in L\L\cap L{\sf D}$ we have
\begin{align*}
{\rm dist}(x, L\L^\complement)\le \min\big\{
|\Phi_L(x') - x''|, \ {\rm dist}(x, {\widetilde{\sf D}}^\complement) \big\} = |\Phi_L(x') - x''|\,,
\end{align*}
and 
\begin{align*}
{\rm dist}(x, L\L^\complement)\ge \min\big\{
\langle M\rangle^{-1}|x''-\Phi_L(x')|,\ 
{\rm dist}(x, {\widetilde{\sf D}}^\complement) \big\} = 
\langle M\rangle^{-1}|x''-\Phi_L(x')|\,,
\end{align*}
so that 
\begin{align}\label{rjk:eq}
R_{jk}:= {\rm dist} (z_{jk}, L\L^\complement)\ge \langle M\rangle^{-1} \big(k/2 - \Phi_L(j/2)\big) \ge 2\,,\ (j, k)\in I_1\,.
\end{align}
Let $(\varphi_{jk})_{jk}\subset \mathsf C^\infty_0(\mathbb R^2)$ be a partition of unity subordinate to the constructed covering. In the following, we use a superposed hat for the (bounded) multiplication operator $\widehat\varphi$ on $\text L^2(\R^2)$ uniquely corresponding to $\varphi\in \mathsf C^\infty_0(\mathbb R^2)$. We estimate individually the quasi-norms

\begin{align*}
\|\widehat\varphi_{jk} \mathds{1}_{L\L}\text P_\ell (\mathds{1}-\mathds{1}_{L\L})\|_p\,, 
\end{align*}
for $p\in (0, 1]$. Let us consider firstly the set $I_2$. For all $(j, k)\in I_2$ we have  

\begin{align*}
\|\widehat\varphi_{jk} \mathds{1}_{L\L}\text P_\ell (\mathds{1}-\mathds{1}_{L\L})\|_p\le \|\widehat\varphi_{jk} \mathds{1}_{L\L}\text P_\ell\|_p \le \|\mathds{1}_{{\sf D}_{jk}}\text P_\ell\|_p\,.
\end{align*}
The first inequality follows from \eqref{holder}. The second inequality holds since 
$\varphi_{jk} 1_{L\L} \le 1_{{\sf D}_{jk}}$ and hence 
$\widehat\varphi_{jk} \mathds{1}_{L\L}\text P_\ell \mathds{1}_{L\L}
\widehat\varphi_{jk} \le \mathds{1}_{{\sf D}_{jk}} \text P_\ell \mathds{1}_{{\sf D}_{jk}}$. 
Using the standard unitary equivalence of the Hamiltonian \eqref{landau} under ``magnetic'' translations, we conclude that the right-hand side coincides with $\|\mathds{1}_{\sf D} \text P_\ell\|_p$, where ${\sf D} = {\sf D}(0, 1)$, as before. By Lemma \ref{sp:lem}, this norm is bounded for all $p>0$, and hence 
\begin{align*}
\|\widehat\varphi_{jk} \mathds{1}_{L\L}\text P_\ell (\mathds{1}-\mathds{1}_{L\L})\|_p\le C 
\end{align*} 
uniformly in $j$ and $k$. 
Applying the $p$-triangle inequality \eqref{p-tri:eq} for quasi-norms, we get 
\begin{align*}
\sum_{(j,k)\in I_2} \|\widehat \varphi_{jk} \mathds{1}_{L\L}\text P_\ell (\mathds{1}-\mathds{1}_{L\L})\|^p_p
\le C|I_2|\le CL\,.
\end{align*}
Suppose now that $(j, k)\in I_1$. Thus
\begin{align*}
\|\widehat\varphi_{jk} \text P_\ell (\mathds{1}-\mathds{1}_{L\L})\|_p\le \|\mathds{1}_{{\sf D}_{jk}}\text P_\ell (\mathds{1}-\mathds{1}_{{\sf D}(z_{jk}, R_{jk})})\|_p
\end{align*}
with $R_{jk} = {\rm dist}(z_{jk}, L{\L^\complement})\ge 2$. By the translation argument, the right-hand side coincides with 
$\| {\rm T}_{1, R_{jk}}\|_p$, where the operator ${\rm T}_{1, R}$ is defined in \eqref{trr:eq}. Consequently, by \eqref{sp:eq}, 
\begin{align*}
\|\widehat\varphi_{jk} \text P_\ell (\mathds{1}-\mathds{1}_{L\L})\|_p\le C\,\exp\big(-(R_{jk}-1)^2/8\big)\,.
\end{align*}
Using the $p$-triangle inequality again, we obtain that
\begin{align}\label{sumkj:eq}
\sum_{(j, k)\in I_1}\|\widehat\varphi_{jk} \text P_\ell (\mathds{1}-\mathds{1}_{L\L})\|^p_p\le C \sum_{(j,k)\in I_1} \exp\big(-p (R_{jk}-1)^2/8\big)\,.
\end{align}
Employing \eqref{rjk:eq}, for any fixed $j$, the summation over $k$ yields the estimate
\begin{align*}
\sum_{(j,k)\in I_1, \,j \ \rm{fixed}} \exp\big(-p(R_{jk}-1)^2/8\big)\le \sum_{k\in\mathbb Z} \e^{-c k^2} = C\,.
\end{align*}
Since $|j|\le 2L$, the right-hand side of \eqref{sumkj:eq} does not exceed $CL$. Putting these estimates together we obtain
\begin{align}\label{prom:eq}
\|\mathds{1}_{L{\sf D}} \mathds{1}_{L\L} \text P_\ell (\mathds{1}-\mathds{1}_{L\L})\|_p^p
\le \sum_{(j,k)\in I_1\cup I_2} \|\widehat\varphi_{jk} \text P_\ell (\mathds{1}-\mathds{1}_{L\L})\|^p_p\le CL\,.
\end{align}

\underline{Case (2):} We fix one disk ${\sf D} = {\sf D}(x_k, r_k)$ such that ${\rm dist}({\sf D}, \L^\complement) \ge c$, so that ${\rm dist}(L{\sf D}, L\L^\complement) \ge cL$. We cover $L{\sf D}$ by unit disks ${\sf D}_j$, $j=1,\ldots,N$ with $N\le CL^2$ and ${\rm dist}({\sf D}_j, L\L^\complement)\ge cL$. As in the proof of Case (1), we introduce a smooth partition of unity $(\varphi_j)_j\subset \mathsf C^\infty_0(\mathbb R^2)$ subordinate to this covering, and estimate:
\begin{align*}
\|\widehat\varphi_j \text P_\ell (\mathds{1}-\mathds{1}_{L\L})\|_p^p\le C \e^{-cL^2}\,.
\end{align*} 
Consequently, by the $p$-triangle inequality \eqref{p-tri:eq}
  
\begin{align}\label{ins:eq}
\|\mathds{1}_{L{\sf D}}\text P_\ell (\mathds{1}-\mathds{1}_{L\L})\|_p^p\le \sum_{j}\|\widehat\varphi_{j} \text P_\ell (\mathds{1}-\mathds{1}_{L\L})\|^p_p
\le C \sum_{j} \e^{-c L^2}\le CL^2 \e^{-cL^2}\le C \e^{-c'L^2}\,.
\end{align}
To complete the proof we add the estimates of the form \eqref{prom:eq} and \eqref{ins:eq}  
for all disks covering $\L$, using the $p$-triangle inequality again. 
\end{proof}

\subsection{Proofs of Lemma \ref{lem:finite} and Theorem \ref{entropy of ground state}}\label{proof}
We use the bound 
\begin{align}\label{ha:eq}
0\le h_\a(t)\le C_\a t^\b(1-t)^\b\,,\quad t\in [0, 1]\,,
\end{align} 
with a positive constant $C_\a<\infty$. Here we choose $\b = \a$ if $\a < 1$, any $\b\in(0,1)$ if $\a = 1$, 
and $\b =1$ if $\a >1$. Since $\L$ is bounded, we have $\L\subset {\sf D}(0, r)$ with some $r >0$. This leads to the estimate
\[
S_\alpha(\Lambda) \le C_\a \tr [\text P_{\le\nu}(\L)^{\b}(\mathds{1}-\text P_{\le \nu}(\L))^{\b}]
\le C_\a \|\mathds{1}_{\L} \text P_{\le\nu}(\mathds{1}-\mathds{1}_\L)\|_\b^{\b}
\le C_\a \|\mathds{1}_{r} \text P_{\le\nu}\|_\b^{\b}\,.
\]
Using the $\b$-triangle inequality and Lemma \ref{sp:lem}, we obtain the estimate (see \eqref{trr:eq} for the definition of $\text T_{r, 0}^{(\ell)}$)
\begin{align*}
S_\alpha(\Lambda)\le C_\a\sum_{\ell =0}^{\nu}                      
\| \text T_{r, 0}^{(\ell)}\|_\b^{\b} < \infty\,,
\end{align*}
as claimed. This proves Lemma \ref{lem:finite}.
\qed

Now we prove Theorem \ref{entropy of ground state}. Since $h_\a$ satisfies the bound \eqref{ha:eq}, it follows from Lemma \ref{mleell:lem} that the coefficient ${\sf M}_{\le \nu}(h_\a)$ is finite. The positivity of ${\sf M}_{\le \nu}(h_\a)$ is obvious from $h_\a\ge0, \mathcal K_{\nu,\xi}\ge 0$, and $h_\a(1)=0$.
If $\a >1$, then the function $h_\a$ is $\mathsf C^1$-smooth on ${[0,1}]$, and hence the claim follows immediately from Theorem \ref{thm:ground state smooth f}. 

In order to prove \eqref{main} for $\a\le 1$ we follow the method of \cite{LSS1}. For $\varepsilon >0$ we choose a smooth function $\zeta_\varepsilon$ such that $0\leq \zeta_\varepsilon\leq 1$ and 
\begin{align*}
\zeta_\varepsilon(t)
=\begin{cases}
1 &\mbox{ if }t\in [0,\varepsilon/2]\cup [1-\varepsilon/2,1]\,,\\
0 &\mbox{ if }t\in [\varepsilon,1-\varepsilon]\,.
\end{cases} 
\end{align*}
In view of estimate \eqref{ha:eq}, we have 
\begin{align}\label{zetah:eq}
|(\zeta_\varepsilon h_\a)(t)| 
\le C \varepsilon^r [t(1-t)]^r,\ \quad r := \frac{\b}{2}\,,
\end{align}
and hence
\begin{align}\label{intermediate}
\|(\zeta_\varepsilon h_\a)( \text P_{\le \nu} (L\L))\|_1
\le C \varepsilon^r
\|\text P_{\le \nu}(L\L)\big(\mathds{1}-\text P_{\le \nu}(L\L)\big)\|_r^r
= C\varepsilon^r \|\mathds{1}_{L\L}\text P_{\le \nu} \big(\mathds{1}-\mathds{1}_{L\L}\big)\|_{2r}^{2r}\,.
\end{align}
By Theorem \ref{Hank:thm}, the quasi-norm on the right-hand side does not exceed 
$CL$. Consequently, 
\begin{align}\label{esttrzetaepsh}
\big|\tr\big[(\zeta_\varepsilon h_\a)
(\text P_{\le \nu}(L\L))\big]\big|\le C\varepsilon^r L\,.
\end{align}
On the other hand, the function $(1-\zeta_\varepsilon)h_\a$ 
vanishes in a vicinity of $0$ and $1$ and, therefore, by Theorem \ref{thm:ground state smooth f},  
we have 
\begin{align}\label{esttrheps}
\tr\big[(1-\zeta_\varepsilon)h_\a (\text P_{\le \nu}(L\L))\big] 
= L\sqrt B |\p\L| {\sf M}_{\le \nu}((1-\zeta_\varepsilon)h_\a) + o(L)\,, \ L\to\infty\,.
\end{align}
By \eqref{zetah:eq}, 
\begin{align}\label{estwheps}
{\sf M}_{\le \nu}(h_\a)-{\sf M}_{\le \nu}((1-\zeta_\varepsilon)h_\a)
= {\sf M}_{\le \nu}(\zeta_\varepsilon h_\a)
\le C\varepsilon^r {\sf M}_{\le \nu}(\tilde{g}^r)\,,\quad \tilde{g}(t) = t(1-t)\,.
\end{align} 
Combining 
\eqref{esttrzetaepsh}, \eqref{esttrheps}, and \eqref{estwheps} gives
\begin{align*}
\limsup\limits_{L\to\infty}\Big|\frac{\tr h_{\alpha }(\text P_{\le \nu}(L\L))}{L}-
\sqrt B |\p\L|{\sf M}_{\le \nu}(h_\a)\Big|\le C\varepsilon^r\,.
\end{align*}
Since $\varepsilon>0$ is arbitrary, this yields the claim. 
\qed

\section{On an improvement to sub-leading terms}

Without going into details, we show here how one can improve the asymptotic expansions in Lemmata \ref{lem:moments} and \ref {lem:general moments} to 
\begin{align} \label{2 term}
\tr f(\text P_\ell(L\L)) &= L^2B\, \frac{|\L|}{2\pi}\,f(1) + L\sqrt{B} \, |\p\L|\, {\sf M}_\ell(f) \, + o(1)\,,
\\\label{gen: 2 term}
\tr f(\text P_{\le n}(L\L)) &= L^2B\,\frac{|\L|}{2\pi}\,(n+1)f(1) + L\sqrt{B} \, |\p\L|\, {\sf M}_{\le n}(f)\, + o(1)\, ,
\end{align}
where $f(t) = t^n, n = 1, 2, \dots$. This follows again by proving the corresponding statements for all natural powers of $\text P_\ell(L\L)$. Here, we use the expansion of Roccaforte in Proposition \ref{Roccaforte} up to the second order $\varepsilon^2$. For any vector $y =  -z \mathsf{J}n_x + tn_x\in\R^2$, written in the form \eqref{zt coordinates}, we have $\|y\|^2 - 2\,\langle y| n_x\rangle^2 = z^2 - t^2$. Then the $\mathcal{O}(\sum_j |y_j|^2)$-term in \eqref{20} takes the form
\begin{align} \label{58} 
\sum_{q=1}^{m-1}\int_{\mathcal S_q} \text dA(x) \,\kappa(x)&\int_{\R^{m-1}}\text d\mathbf{z} \int_{\R^{m-1}}\text d\mathbf{t} \,\Big[\big(\sum_{j=1}^q z_j\big)^2 -\big(\sum_{j=1}^q t_j\big)^2\Big] \,\exp(\mathrm{i}\mfr{B}{2}\langle \mathbf{z}|\mathsf{S}\mathbf{t}\rangle)\notag
\\&\times\,\prod_{j=1}^{m} \mathcal{L}_\ell((z_j^2+t_j^2)/2)g(z_j) g(t_j) + o(1)\,.
\end{align}
If we exchange the $\mathbf z$ and $\mathbf t$ variables, the integrand is seen to be almost anti-symmetric except for the sign in the exponent. This can be remedied by changing, for instance, $\mathbf t$ to $-\mathbf t$. Hence, the integral in \eqref{58} vanishes by symmetry.

\medskip
At present we do not know how to extend the expansions \eqref{2 term} and \eqref{gen: 2 term} to the entropy function $f=h_\a$, but believe that the corresponding term vanishes also in this case. In other words, there is zero topological (entanglement) entropy, see \cite{KP,LW}. 

\begin{appendix}

\section{Roccaforte's formula for the area of intersections}\label{App:R}

We recall Roccaforte's results in \cite[Corollary 2.4, Lemma 2.5]{R} for the special case of the Euclidean plane. Since the boundary curve $\p\L$ is a one-dimensional manifold, the second fundamental form of $\p\L$ is just its curvature. Therefore, his formula takes the simpler form given in Proposition \ref{Roccaforte} below.

In our application we scale out $\varepsilon$ and identify $L\sqrt{B}=1/\varepsilon$. For given points/vectors $v_1,\ldots,v_r$ in $\R^2$ we denote by $\L_\varepsilon := \L\cap(\L+\varepsilon v_1) \cap \cdots\cap(\L+\varepsilon v_r)$ the intersection of $\L$ with its $r$ translates.

\begin{prop} \label{Roccaforte} 
Let $\L\subset\R^2$ be a bounded $\mathsf C^3$-region. Then for arbitrary $(v_1,\ldots,v_r)\in\R^{2}\times\cdots\times\R^2$, we have
\begin{align} 
\bigg||\L\setminus\L_\varepsilon| -\varepsilon\int_{\p\L} \mathrm{d} A(x)\, \max\big\{0,\langle v_1| n_x\rangle,\ldots,\langle v_r|n_x\rangle \big\}\bigg|\le C\varepsilon^2 \sum_{j=1}^r \|v_j\|^2,
\end{align}
with some constant $C$ independent of $\varepsilon>0$ or $v_1, v_2, \dots, v_r$, 
where $n_x$ is the inward unit normal to the curve $\p\L$. 

Moreover, for all $(v_1,\ldots,v_r)\in\R^{2}\times\cdots\times\R^2$ except for a set 
of $2r$-dimensional Lebesgue measure zero, we have
\bea |\L\setminus\L_\varepsilon| &=& \varepsilon\int_{\p\L} \mathrm{d} A(x)\, \max\big\{0,\langle v_1|n_x\rangle,\ldots,\langle v_r| n_x\rangle\big\}\notag
\\
&+&\mfr12 \,\varepsilon^2\, \sum_{q=1}^r \int_{\mathcal C_q} \mathrm{d}A(x)\, \kappa(x)\big[\|v_{q}\|^2 - 2\,\langle v_{q}|n_x\rangle^2\big] + o(\varepsilon^2)\,,\quad  \varepsilon\to 0\,.
\eea 
Here, $\mathcal C_q := \big\{x\in\p\L:\langle v_q|n_x\rangle = \max\big\{0,\langle v_1|n_x\rangle,\ldots,\langle v_r|n_x\rangle\big\}\big\}$, and $\kappa(x)$ is the curvature of $\p\L$ at the point $x\in\p\L$. 
\end{prop}

\section{Miscellaneous identities}

\subsection{Proof of the result \texorpdfstring{\eqref{exponent} 
of a change of variables}{LG}}\label{Misc}

The term quadratic in $\xi$ is obviously correct and we start with the linear 
term in $\xi$. The inverse of the matrix  $\mathsf{A} := \mathsf{A}^{(q)}$ has the entries
\be \label{A inverse} \mathsf{A}^{-1}(i,j) = \left\{\begin{array}{rl}1&\mbox{ if } i = j\\-1&\mbox{ if } 1\le i=j-1\le q-1 \\
-1&\mbox{ if } q+1\le j=i-1\le m-2\\0&\mbox{ otherwise }\end{array}\right.
\ee
and
\begin{align*} \sum_{i=1}^{m-1} T_i &= \sum_{j=1}^{m-1} t_j \sum_{i=1}^{m-1} \mathsf{S}_{ij}
\\
&=\sum_{j=1}^{m-1} \big(\mathsf{A}^{-1}\boldsymbol \tau)_j (m-2j)
\\
&=\sum_{j=1}^{m-1}(m-2j)\tau_j - \sum_{j=1}^{q-1}(m-2j)\tau_{j+1} - \sum_{j=q+2}^{m-1}(m-2j)\tau_{j-1}
\\
&=(m-2) \tau_1 - 2\sum_{j=2}^q\tau_j + 2\sum_{j=q+1}^{m-2}\tau_{j} - (m-2)\tau_{m-1}
\\
&\rightsquigarrow (m-2) \tau_1 - 2\sum_{j=2}^q\tau_j - 2\sum_{j=q+1}^{m-2}\tau_{j} + (m-2)\tau_{m-1}
\\
&=(m-2)(\tau_1+\tau_{m-1}) -  2\sum_{j=2}^{m-2}\tau_j \,.
\end{align*}
In the $\rightsquigarrow$ line we have reversed the signs of $\tau_{q+1},\ldots,\tau_{m-1}$. Now we replace $\xi$ on the left-hand side of \eqref{exponent} by $\xi - (\tau_1+\tau_{m-1})/2$, which finally yields the linear term 
\begin{align*} -m\xi(\tau_1+\tau_{m-1}) + \xi \Big[(m-2) (\tau_1+\tau_{m-1}) - 2\sum_{j=2}^{m-2}\tau_j \Big] = -2\xi\sum_{j=1}^{m-1}\tau_j\,.
\end{align*}
Similarly, since 
$$ \sum_{1\le j\le m-1} \mathsf{S}_{jk}\mathsf{S}_{j\ell} = \left\{\begin{array}{ll}m-2&\mbox{ if } k=\ell\\m-1-2|k-\ell|&\mbox{ if } k\not=\ell\end{array}\right.\,,
$$
we have ($T_m=0$)
\begin{align*} \sum_{1\le j\le m-1} T_j^2 &+ \sum_{1\le j\le m} t_j^2 = \sum_{1\le k,\ell\le m-1} t_k t_\ell \sum_{1\le j\le m-1} \mathsf{S}_{jk} \mathsf{S}_{j\ell} + \sum_{1\le k\le m-1} t_k^2 + t_m^2
\\
&=(m-1)\sum_{1\le k\le m-1} t_k^2  +  \Big(\sum_{1\le k\le m-1}t_k\Big)^2 + 2 \sum_{1\le k<\ell\le m-1} t_k t_\ell (m-1-2(\ell-k))
\\
&=m\Big(\sum_{1\le k\le m-1}t_k\Big)^2 - 4 \sum_{1\le k<\ell\le m-1} t_k t_\ell (\ell-k)
\\
&=m(\tau_1+\tau_{m-1})^2 - 4 \sum_{1\le k<\ell\le m-1} (\mathsf{A}^{-1}\boldsymbol\tau)(k) (\mathsf{A}^{-1}\boldsymbol\tau)(\ell) (\ell-k)\,.
\end{align*}
Next, we write
\begin{align*}\sum_{1\le k<\ell\le m-1} (\mathsf{A}^{-1}\boldsymbol \tau)(k) (\mathsf{A}^{-1}\boldsymbol\tau)(\ell) (\ell-k)&=\sum_{1\le k<\ell\le q-1}(\tau_k - \tau_{k+1})(\tau_\ell - \tau_{\ell+1})(\ell-k)
\\
&+\sum_{1\le k\le q-1} (\tau_k-\tau_{k+1})\tau_q (q-k)
\\
&+\sum_{1\le k\le q-1} (\tau_k-\tau_{k+1})\tau_{q+1} (q+1-k) + \tau_q \tau_{q+1}
\\
&+\sum_{1\le k\le q-1,q+2\le \ell} (\tau_k-\tau_{k+1})(\tau_\ell - \tau_{\ell-1})(\ell - k)
\\
&+\sum_{q+2\le k < \ell \le m-1} (\tau_k-\tau_{k-1})(\tau_\ell - \tau_{\ell-1})(\ell - k)
\\
&+\tau_q \tau_{q+1} + \tau_q \sum_{q+2\le\ell\le m-1}(\tau_\ell - \tau_{\ell-1})(\ell - q)
\\
&+\tau_{q+1} \sum_{q+2\le\ell\le m-1}(\tau_\ell - \tau_{\ell-1})(\ell - q -1)\,.
\end{align*}
Now we reverse the signs of $\tau_{q+1},\ldots,\tau_{m-1}$. To this end, let $\mathsf{I}:=\mathsf{I}^{(q)} = \mathrm{diag}(\underbrace{1,\ldots,1}_q,\underbrace{-1,\ldots,-1}_{m-1-q})$ be the diagonal $(m-1)\times(m-1)$ matrix that provides this reversal. Then we get
\begin{align}\sum_{1\le k<\ell\le m-1} (\mathsf{I} \mathsf{A}^{-1}\boldsymbol \tau)(k) (\mathsf{I}\mathsf{A}^{-1}\boldsymbol\tau)(\ell) (\ell-k)&=\sum_{1\le k<\ell\le q-1}(\tau_k - \tau_{k+1})(\tau_\ell - \tau_{\ell+1})(\ell-k)
\label{A}\\
&+\tau_q\Big(\tau_1(q-1) - \sum_{2\le k\le q} \tau_k \Big)
\label{B}\\
&-\tau_{q+1} \Big(\tau_1 q-\sum_{2\le k\le q} \tau_k\Big)
\label{C}\\
&-\sum_{1\le k\le q-1,q+2\le \ell} (\tau_k-\tau_{k+1})(\tau_\ell - \tau_{\ell-1})(\ell - k)
\label{D}\\
&+\sum_{q+2\le k < \ell \le m-1} (\tau_k-\tau_{k-1})(\tau_\ell - \tau_{\ell-1})(\ell - k)
\label{E}\\
&-\tau_q \tau_{q+1} - \tau_q \sum_{q+2\le\ell\le m-1}(\tau_\ell - \tau_{\ell-1})(\ell - q)
\label{F}\\
&+\tau_{q+1} \sum_{q+2\le\ell\le m-1}(\tau_\ell - \tau_{\ell-1})(\ell - q -1)\,.
\label{G}
\end{align}
Recall that $\xi$ in \eqref{exponent} is replaced by $\xi-(\tau_1+\tau_{m-1})/2$. That is, on top of the term $\mfr14\sum_{1\le j\le m} (T_j^2 + t_j^2)$ (after reversing signs of $\tau_{q+1},\ldots,\tau_{m-1}$) we also have the term
\begin{align*} m(\mfr12(\tau_1+\tau_{m-1}))^2  &- \mfr12(\tau_1+\tau_{m-1}) \Big[(m-2)(\tau_1+\tau_{m-1}) - 2\sum_{2\le j\le m-2} \tau_j\Big] 
\\&= (\tau_1+\tau_{m-1})^2(1-\mfr{m}{4}) + (\tau_1+\tau_{m-1})\sum_{2\le j\le m-2} \tau_j\,.
\end{align*}
Let $\mathsf{B}$ be the $(m-1)\times (m-1)$ matrix defined as 
\begin{align*} \langle \boldsymbol \tau,\mathsf{B}\boldsymbol \tau\rangle &:= \mfr{m}{4}(\tau_1-\tau_{m-1})^2 - \sum_{1\le k<\ell\le m-1} (\mathsf{I} \mathsf{A}^{-1}\boldsymbol \tau)(k) (\mathsf{I} \mathsf{A}^{-1}\boldsymbol\tau)(\ell) (\ell-k) 
\\
&+ (\tau_1+\tau_{m-1})^2(1-\mfr{m}{4}) + (\tau_1+\tau_{m-1})\sum_{2\le j\le m-2} \tau_j\,.
\end{align*}
Then we need to show that $\mathsf{B} = \mathbb 1$. We distinguish between certain ranges of indices. 
\begin{itemize}
\item $i_0 = j_0 = 1$: From the first and last term in the definition of $\mathsf{B}$ we get $\mathsf{B}_{1,1} = \mfr{m}{4} + (1-\mfr{m}{4}) = 1  \ \checkmark$
\item $1=i_0 < j_0\le q-1$: choose in \eqref{A} $k=1,\ell=j_0$ or $k=1,\ell=j_0-1$. Then we have $\mathsf{B}_{1,j_0} = 1-(j_0-i_0) + (j_0-1-i_0) = 0 \ \checkmark$
\item $i_0=1, j_0=q$: choose in \eqref{A} $k=1,\ell=q-1$. Then, together with \eqref{B} and the last term in the definition of $\mathsf{B}$ we get $\mathsf{B}_{1,q} = q-2 - (q-1) +1 = 0 \ \checkmark$
\item $i_0=1, j_0=q+1$: choose in \eqref{D} $k=1,\ell=q+2$. Then, together with \eqref{C} and the last term in the definition of $\mathsf{B}$ we get $\mathsf{B}_{1,q+1} = q -(q+1) -1= 0\ \checkmark$
\item $i_0=1, q+2\le j_0\le m-1$: choose in \eqref{D} $k=1,\ell=j_0$ and $k=1,\ell=j_0+1$. Then, together with the last term in the definition of $\mathsf{B}$ we get $\mathsf{B}_{1,j_0} = j_0-1 -j_0 +1 = 0 \ \checkmark$
\item $2\le i_0 = j_0\le q-1$: choose in \eqref{A} $k=i_0-1,\ell=i_0$. Then $\mathsf{B}_{i_0,i_0} = 1 \ \checkmark$
\item $2\le i_0<j_0\le q-1$: choose in \eqref{A} $k=i_0,\ell=j_0$, $k=i_0-1,\ell=j_0-1$, $k=i_0,\ell=j_0-1$, or $k=i_0-1,\ell=j_0$. Then $\mathsf{B}_{i_0,j_0} = -(j_0 - i_0) - (j_0 - i_0) + (j_0-1-i_0) + (j_0 - i_0 +1) = 0\ \checkmark$
\item $2\le i_0\le q-1, j_0 = q$: choose in \eqref{A} $\ell = q-1$ and $k=i_0$ or $k=i_0-1$ and in \eqref{B} $k=i_0$. Then $\mathsf{B}_{i_0,q} = q-1-i_0-(q-1-i_0+1)+1=0\ \checkmark$
\item $2\le i_0\le q-1, j_0 = q+1$: choose in \eqref{C} $k=i_0$ and in \eqref{D} $\ell=q+2$ and $k=i_0$ or $k=i_0-1$. Then $\mathsf{B}_{i_0,q+1} = -1 -(q+2-i_0) + q+2-i_0+1 = 0\ \checkmark$
\item $2\le i_0\le q-1, q+2\le j_0 \le m-1$: choose in \eqref{D} $k=i_0$ and $\ell=j_0$ or $\ell=j_0+1$, or $k=i_0-1$ and $\ell=j_0$ or $\ell=j_0+1$. Then $\mathsf{B}_{i_0,j_0} = (j_0-i_0) -(j_0+1-i_0)-(j_0-i_0+1)+(j_0+1-i_0+1) = 0\ \checkmark$
\item $i_0=q, j_0=q$: this term appears in \eqref{B} if $k=q$. Then $\mathsf{B}_{q,q} = 1\ \checkmark$
\item $i_0=q, j_0=q+1$: choose in \eqref{C} $k=q$, in \eqref{D} $k=q-1$ and $\ell=q+2$, in \eqref{F} $\ell=q+2$. Then $\mathsf{B}_{q,q+1} = -1 + (q+2-q+1) +1 -(q+2-q) 0\ \checkmark$
\item $i_0=q, q+2\le j_0\le m-1$: choose in \eqref{D} $k=q-1$ and $\ell=j_0$ or $\ell = j_0+1$ and in \eqref{F} $\ell=j_0$ or $\ell=j_0+1$. Then $\mathsf{B}_{q,j_0} = -(j_0-q+1) + (j_0+1-q+1) +(j_0-q) -(j_0+1-q) =  0\ \checkmark$
\item $i_0=j_0=q+1$: choose in \eqref{G} $\ell = q+2$. Then $\mathsf{B}_{q+1,q+1} = 1 \ \checkmark$
\item $i_0=q+1, j_0=q+2$: choose in \eqref{E} $k=q+2,\ell=q+3$ and in \eqref{G} $\ell = q+2$ or $\ell = q+3$. Then $\mathsf{B}_{q+1,q+2} = -1 -1 +2 = 0\ \checkmark$
\item $i_0=q+1, q+2\le j_0\le m-1$: choose in \eqref{E} $k=q+1$ and $\ell=j_0$ or $\ell = j_0+1$, and in \eqref{G} $\ell = j_0$ or $\ell=j_0+1$. Then $\mathsf{B}_{q+1,j_0} = (j_0-q-1) -(j_0+1-q-1) -(j_0-q-1) +(j_0+1-q-1) = 0\ \checkmark$
\item $q+2\le i_0 = j_0\le m-2$: choose in \eqref{E} $k=i_0$ and $\ell=i_0+1$. Then $\mathsf{B}_{i_0,i_0} = 1 \ \checkmark$
\item $i_0 = j_0 = m-1$: $\mathsf{B}_{m-1,m-1}$ comes from the first and third one in the definition of $\mathsf{B}$. That is, $\mathsf{B}_{m-1,m-1} = \mfr{m}{4} + (1-\mfr{m}{4}) = 1 \ \checkmark$
\item $q+2\le i_0<j_0\le m-2$: choose in \eqref{E} $k=i_0,\ell=j_0$, $k=i_0+1,\ell=j_0$, $k=i_0,\ell=j_0+1$, $k=i_0+1,\ell=j_0+1$. Then $\mathsf{B}_{i_0,j_0} = -(j_0-i_0) + (j_0-i_0-1) + (j_0+1-i_0) -(j_0-i_0) = 0\ \checkmark$
\item $q+2\le i_0, j_0= m-1$: choose \eqref{E} $\ell=m-1$ and $k=i_0$ or $k=i_0+1$. In conjunction with the last term in the definition we obtain $\mathsf{B}_{i_0,m-1} = -(m-1-i_0) + (m-1-i_0-1) + 1=0\ \checkmark$
\end{itemize}

Finally, switching $\xi$ to $-\xi$ proves \eqref{exponent}.

\subsection{Change of variables in Laguerre polynomials}\label{appendix B.2}

With $T_j$ from \eqref{def T} and $\boldsymbol \tau = \mathsf{A} \mathbf t, \mathsf{A}=\mathsf{A}^{(q)}$ from \eqref{matrix A}, we claim that
\be T_j = \left\{\begin{array}{ll} \tau_1 - \tau_j - \tau_{j+1} -\tau_{m-1}&\mbox{ if } 1\le j\le q-1\\\tau_1 - \tau_q - \tau_{m-1}&\mbox{ if } j=q\\\tau_1 + \tau_{q+1} - \tau_{m-1}&\mbox{ if } j=q+1\\\tau_1+\tau_{j-1}+\tau_j-\tau_{m-1}&\mbox{ if } q+2\le j\le m-1\end{array}\right.\,.
\ee
To this end, we write
\begin{align*} T_j&=\sum_{1\le k\le j-1} (\mathsf{A}^{-1}\boldsymbol \tau)(k) - \sum_{j+1\le k\le m-1} (\mathsf{A}^{-1}\boldsymbol \tau)(k)
\end{align*}
and assume for example that $1\le j\le q-1$. Then, using \eqref{A inverse}
\begin{align*} T_j&=\sum_{1\le k\le j-1} (\tau_k - \tau_{k+1}) - \sum_{j+1\le k\le q+1} (\mathsf{A}^{-1}\boldsymbol \tau)(k) - \sum_{j+1\le k\le q+1} (\mathsf{A}^{-1}\boldsymbol \tau)(k)
\\
&=\tau_1-\tau_j - \sum_{j+1\le k\le q-1} (\tau_k - \tau_{k+1}) - \tau_q - \tau_{q+1} - (\tau_{m-1} - \tau_{q+1})
\\
&=\tau_1 - \tau_j - \tau_{j+1} -\tau_{m-1}\,.
\end{align*}
Also, 
\begin{align*} T_q &= \sum_{1\le k\le q-1} (\tau_k - \tau_{k+1}) - \tau_{q+1}  - \sum_{q+1\le k\le m-1} (\tau_k-\tau_{k-1})
\\
&=\tau_1-\tau_q-\tau_{m-1}\,.
\end{align*}
The other two cases follow in a similar vein. After reversing the signs of $\tau_{q+1},\ldots,\tau_{m-1}$, 
\be T_j \rightsquigarrow \left\{\begin{array}{ll} \tau_1 - \tau_j - \tau_{j+1} + \tau_{m-1}&\mbox{ if } 1\le j\le q-1\\\tau_1 - \tau_q + \tau_{m-1}&\mbox{ if } j=q\\\tau_1 - \tau_{q+1} + \tau_{m-1}&\mbox{ if } j=q+1\\\tau_1-\tau_{j-1}-\tau_j+\tau_{m-1}&\mbox{ if } q+2\le j\le m-1\end{array}\right.\,.
\ee
Next, we replace $\xi$ by $\xi-(\tau_1+\tau_{m-1})/2$. So we subtract from $T_j$ the term $(\tau_1+\tau_{m-1})/2$. This leads to
\be T_j \rightsquigarrow \widetilde{T}_j := \left\{\begin{array}{ll} - \tau_j - \tau_{j+1} &\mbox{ if } 1\le j\le q-1\\- \tau_q &\mbox{ if } j=q\\- \tau_{q+1} &\mbox{ if } j=q+1\\-\tau_{j-1}-\tau_j&\mbox{ if } q+2\le j\le m-1\end{array}\right.\,.
\ee
Let 
\be t_j \rightsquigarrow \widetilde{t}_j := (\mathsf{I}\mathsf{A}^{-1}\boldsymbol \tau)(j) = \left\{\begin{array}{ll} \tau_j - \tau_{j+1} &\mbox{ if } 1\le j\le q-1\\\tau_q &\mbox{ if } j=q\\- \tau_{q+1} &\mbox{ if } j=q+1\\\tau_{j-1}-\tau_j&\mbox{ if } q+2\le j\le m-1\end{array}\right.\,.
\ee
Now we replace $\xi$ by $-\xi$ so that
\begin{align} (\omega+\mathrm{i}(2\xi+T_j))^2 + t_j^2 &\rightsquigarrow (\omega-\mathrm{i}(2\xi+\widetilde{T}_j))^2 + \widetilde{t}_j^2\,.
\end{align}
For $j=q$ the last expression equals
\[ (\omega-\mathrm{i}(2\xi+\tau_q))^2 + \tau_q^2 =\omega^2 -2\mathrm{i}\omega(2\xi+\tau_q) - (2\xi)^2 -4\xi\tau_q\,.
\]
Next, we shift the integration with respect to $\tau_q$ by $\xi$. That is, we replace $\tau_q$ by $\tau_{q}-\xi$. Then we have
\begin{align*} 
\omega^2 -2\mathrm{i}\omega(2\xi+\tau_q) - (2\xi)^2 -4\xi\tau_q &\rightsquigarrow \omega^2 -2\mathrm{i}\omega(\xi+\tau_q) - (2\xi)^2 -4\xi(\tau_q-\xi)
\\
&=\omega^2 -2\mathrm{i}\omega(\xi+\tau_q) - 4\xi\tau_q
\\
&=(\omega-2\mathrm{i}\xi)(\omega-2\mathrm{i}\tau_q)\,.
\end{align*}
For $j=q+1$ we get in the end the expression $(\omega-2\mathrm{i}\xi)(\omega-2\mathrm{i}\tau_{q+1})$. 

For $1\le j\le q-1$ we have
\begin{align*} 
(\omega+\mathrm{i}(2\xi - \tau_j - \tau_{j+1}))^2 + (\tau_j - \tau_{j+1})^2 &\rightsquigarrow (\omega-\mathrm{i}(2\xi + \tau_j + \tau_{j+1}))^2 + (\tau_j - \tau_{j+1})^2
\\
&\rightsquigarrow (\omega-\mathrm{i}(\tau_j + \tau_{j+1}))^2 + (\tau_j - \tau_{j+1})^2
\\
&=(\omega-2\mathrm{i}\tau_j)(\omega-2\mathrm{i}\tau_{j+1})\,.
\end{align*}
Similarly, for $q+2\le j\le m-1$,
\begin{align*} 
(\omega+\mathrm{i}(2\xi - \tau_{j-1} - \tau_{j}))^2 + (\tau_{j-1} - \tau_{j})^2 &\rightsquigarrow (\omega-\mathrm{i}(2\xi + \tau_{j-1} + \tau_{j}))^2 + (\tau_{j-1} - \tau_{j})^2
\\
&\rightsquigarrow (\omega-2\mathrm{i}\tau_{j-1})(\omega-2\mathrm{i}\tau_{j})\,.
\end{align*}
Finally, with $T_m=0$ and $t_m = \tau_1+\tau_{m-1}$,
\begin{align*} 
(\omega+\mathrm{i}(2\xi + T_m))^2 + (\tau_{1} + \tau_{m-1})^2 &\rightsquigarrow (\omega-\mathrm{i}(2\xi - \tau_{1} - \tau_{m-1}))^2 + (\tau_{1} - \tau_{m-1})^2
\\
&\rightsquigarrow (\omega-2\mathrm{i}\tau_{1})(\omega-2\mathrm{i}\tau_{m-1})\,.
\end{align*}

\subsection{Proof of the identity \texorpdfstring{\eqref{Hermite identity}}{LG}}\label{remarkable}

We start out with the representation of the Laguerre polynomial as a contour integral in the complex plane $\mathbb C$, namely
\[\mathcal{L}_\ell(x) = \frac1{2\pi\mathrm{i}}\oint_\Gamma \frac{\text d t}{(1-t)t^{\ell+1}}\, \exp[-xt/(1-t)]\,.
\]
Here, the contour $\Gamma$ is, say, a circle of radius $<1$, centered at the origin $0$, and with counter-clockwise orientation. Then, for any given pair $\xi,\tau\in\R$, we have
\begin{align*} \frac{1}{\sqrt{2\pi}} &\int_\R \text d \omega \, \mathcal{L}_\ell\big((\omega-2\mathrm{i}\xi)(\omega-2\mathrm{i}\tau)/2\big) \exp(-\omega^2/4) 
\\
& = \frac1{2\pi\mathrm{i}} \oint_\Gamma\frac{\text d t}{(1-t)t^{\ell+1}} \frac{1}{\sqrt{2\pi}}\int_\R \text d \omega \, \exp\big[-\omega^2/4 -(\omega-2\mathrm{i}\xi)(\omega-2\mathrm{i}\tau)t/(2(1-t))\big]\,.
\end{align*}
Now we observe
\begin{align*} -\mfr14 \omega^2 -\frac{(\omega-2\mathrm{i}\xi)(\omega-2\mathrm{i}\tau)t}{2(1-t)} &= -\frac{1+t}{4(1-t)}\Big(\omega-\frac{2\mathrm{i}(\xi+\tau)t}{1+t}\Big)^2 + \frac{2\xi\tau t}{1-t} - \frac{t^2(\xi+\tau)^2}{1-t^2}
\end{align*}
and perform the $\omega$-integration. This yields
$$ \Big(\frac{4\pi(1-t)}{1+t}\Big)^{1/2} \exp\Big[ \frac{2\xi\tau t}{1-t} - \frac{t^2(\xi+\tau)^2}{1-t^2}\Big]\,.
$$
By the Cauchy integral formula (of the year 1831) we therefore get
\begin{align*} \frac{1}{\sqrt{2\pi}} &\int_\R \text d\omega \, \mathcal{L}_\ell\big((\omega-2\mathrm{i}\xi)(\omega-2\mathrm{i}\tau)/2\big) \exp(-\omega^2/4) 
\\
& = \sqrt{2} \,\frac1{2\pi\mathrm{i}}\oint_\Gamma\frac{\text d t}{\sqrt{1-t^2}\,t^{\ell+1}}\,\exp\Big[ \frac{2\xi\tau t}{1-t} - \frac{t^2(\xi+\tau)^2}{1-t^2}\Big]
\\
&=\sqrt{2}\, \frac{1}{\ell!} \, \frac{\text d^\ell}{\text d t^\ell}\Big|_{t=0} \frac{1}{\sqrt{1-t^2}}\,\exp\Big[ \frac{2\xi\tau t}{1-t} - \frac{t^2(\xi+\tau)^2}{1-t^2}\Big]\,.
\end{align*}
The Mehler formula (of the year 1866) 
\be \frac{1}{\sqrt{1-t^2}}\,\exp\Big[ \frac{2\xi\tau t}{1-t} - \frac{t^2(\xi+\tau)^2}{1-t^2}\Big] = \sum_{\ell=0}^\infty \frac{H_\ell(\xi) H_\ell(\tau)}{\ell!} \Big(\frac{t}{2}\Big)^\ell\,,\quad     |t|<1
\ee
now completes the proof of \eqref{Hermite identity}.

\end{appendix}

\end{document}